\documentclass[11pt]{article}

\usepackage[a4paper,margin=1in]{geometry}
\usepackage{amsmath,amssymb,amsfonts}
\usepackage{graphicx}
\usepackage{multirow}
\usepackage{amsthm}
\usepackage{mathrsfs}
\usepackage[title]{appendix}
\usepackage{xcolor}
\usepackage{textcomp}
\usepackage{manyfoot}
\usepackage{booktabs}
\usepackage[colorlinks=true,linkcolor=blue,citecolor=blue,urlcolor=blue]{hyperref}

\newtheorem{thm}{Theorem}[section]
\newtheorem{theorem}[thm]{Theorem}
\newtheorem{lemma}[thm]{Lemma}
\newtheorem{proposition}[thm]{Proposition}
\newtheorem{corollary}[thm]{Corollary}
\theoremstyle{definition}

\newtheorem{remark}[thm]{Remark}
\newtheorem{example}[thm]{Example}

\newtheorem{assumption}[thm]{Assumption}

\begin{document}

\title{Genealogical Expansions of Positive Fredholm Operators via a Reference-Point Method}

\author{
Ryo Oizumi\thanks{Email: \texttt{ooizumi-ryou@ipss.go.jp}}\\
International Relations, National Institute of Population and Social Security Research\\
2-2-3 Uchisaiwai-cho, Chiyoda-ku, Tokyo 100-0011, Japan
\and
Kensaku Kinjo\thanks{Email: \texttt{kinjo@kushiro-ct.ac.jp}}\\
National Institute of Technology, Kushiro College\\
Otanoshike-Nishi 2-32-1, Kushiro-Shi, Hokkaido 084-0916, Japan
\and
Yuki Chino\thanks{Email: \texttt{chino@nycu.edu.tw}}\\
Department of Applied Mathematics, National Yang Ming Chiao Tung University\\
Hsinchu 30010, Taiwan
}
\date{}

\maketitle

\begin{abstract}
We study positive Fredholm integral operators that arise as next-generation
operators in structured population models.  The main problem is to represent
the dominant eigenvalue and the associated right and left eigenfunctions
without using Fredholm determinants or finite-dimensional discretization.  We
introduce a reference-point construction: a rank-one correction on the space
of kernels, determined by a fixed pair \((x_0,y_0)\), which reorganizes
iterated kernels into a renewal-type series.  Under an explicit dominant
spectral separation assumption and a scalar non-resonance condition for the
chosen reference pair, the resulting \(\Gamma_n\)-series converges at the
spectral radius pointwise absolutely and gives the leading eigensystem.  The coefficients also
have a closed combinatorial expression in terms of ordinary partial Bell
polynomials.  For discrete-time integral projection models and for
multi-state McKendrick equations, the same construction yields
Euler--Lotka-type characteristic equations and formulas for demographic
quantities such as stable distributions, reproductive values, type
reproduction numbers, generation intervals, and expected generation numbers.
The resulting genealogical expansion resolves the leading eigensystem into
successive reproductive and transition contributions encoded by the iterated
kernels.
\end{abstract}



\medskip

\section{Introduction}

This is an operator-theoretic paper on the dominant eigensystem of positive
Fredholm integral operators.  Given a positive kernel operator arising from a
renewal or next-generation formulation, our aim is to reconstruct its
dominant eigenvalue and right and left eigenfunctions directly at the kernel
level.  In this sense the paper complements, rather than replaces, existing
population models, numerical discretizations, and renewal-equation
formulations.

Integral Projection Models (IPMs) provide one important source of such
operators in population dynamics with continuous state variables
\cite{ellner2006integral,ellner2016data}.  A broad literature
\cite{white2016fitting,nicole2011interdependent,doak2021critical,
merow2014advancing,elderd2016quantifying,coulson2008dynamics,
coulson2010using} connects their spectral theory with stable states,
generation times, reproductive values, and related demographic quantities.
Renewal and structured-population theory already provide general
formulations and asymptotic frameworks \cite{feller1941integral}.
Diekmann et al.~\cite{diekmann1998formulation}
formulate general deterministic structured-population models in terms of
individual and population states.  Boldin et al.~\cite{boldin2023population}
emphasize the renewal-equation bookkeeping in discrete time: instead of
updating the whole population state directly, the renewal formulation focuses
on newborn individuals and connects naturally with next-generation matrices
and reproduction numbers.  Related discrete-time population dynamics can also
be formulated directly on spaces of measures \cite{thieme2020discrete}.
Franco et al.~\cite{franco2023modelling} treat
measure-valued renewal equations for physiologically structured populations,
prove long-time behaviour under a regularising kernel assumption, and relate
the renewal-equation formulation to PDE models.  Our purpose is different and
more local: we add a determinant-free representation of the dominant
eigensystem inside these frameworks.

The same operator-theoretic issue appears for continuous-time structured
models such as multi-state McKendrick equations
\cite{m1925applications,foerster1959some,perthame2007transport}
and path-integral population models
\cite{oizumi2013optimal,oizumi2022book}, where the renewal kernel encodes
state variables beyond age.

\medskip
\noindent\textbf{Aim and main result.}

Beyond the long-run growth rate and stable structure, the dominant eigensystem
of an IPM describes how initial traits, stages, or age--state conditions
contribute to population-scale outcomes.  These quantities are governed by the
spectral radius and eigenfunctions of a positive Fredholm operator.  Fredholm's
classical determinant formulation \cite{fredholm1903class} is fundamental, but
it does not directly display the multigenerational transition and reproductive
contributions that are central in demographic interpretation.

The aim of this paper is to give an explicit determinant-free representation
of this dominant eigensystem.  The central construction is a reference-point
recursion at the kernel level.  It produces iterates
\(\{\Gamma_n\}_{n\ge1}\) which reorganize the leading Perron contribution of
the iterated kernels into a renewal-type,
genealogical expansion.  The word ``genealogical'' is used in this constructive
sense: the terms are indexed by successive transition or reproductive steps.
This is not meant as an alternative to numerical discretization.  Its purpose
is to display the continuous-state eigenstructure that is obscured once the
kernel has been replaced by a quadrature matrix.

The essential assumptions are dominant spectral separation and reference
non-resonance.  They have different roles: the former isolates the Perron
component, while the latter controls a scalar denominator associated with the
chosen reference pair.  The algebraic reference-point identities do not require
non-resonance; that condition enters only for ordinary absolute convergence at
the Perron root and conclusions based on that convergence.  The other
hypotheses are technical conditions ensuring that the kernel-level construction
is meaningful or that these two assumptions can be verified in concrete models.
Under these assumptions, the main abstract statement,
Theorem~\ref{thm:nonHS_extension}, conditionally extends the coefficient
construction identified on the admissible Hilbert--Schmidt subclass to
non-Hilbert--Schmidt kernels and represents the leading eigenfunction by a
pointwise absolutely convergent \(\Gamma_n\)-series.  The iterates also
admit a closed form in terms of ordinary partial Bell polynomials.

We then apply the framework to two model classes.  For simple discrete-time
IPMs, the stable distribution and reproductive value are expressed as sums of
multistep contributions centered at a reference state.  For multi-state
McKendrick equations, the age-zero renewal kernel
\(\psi(\cdot,\cdot;r)\) leads to an Euler--Lotka-type characteristic equation
and to demographic indicators such as type reproduction numbers, expected
generation numbers, and generation intervals, all expressed directly through
continuous-state kernels without discretization.

\medskip
\noindent\textbf{Organization of the paper.}

Section~2 develops the determinant-free Fredholm formulation using the
reference-point operator and establishes the convergence mechanism of
the $\Gamma_n$-series under dominant spectral separation and reference non-resonance assumptions.
Section~3 applies the construction to discrete-time IPMs and clarifies
its connection with the taboo decomposition in Markov chains.
Section~4 treats the multi-state McKendrick equation and derives
genealogical interpretations and demographic indicators.
Section~5 concludes with a discussion of implications and future
directions.
\section{Reference-point construction for nonnegative Fredholm integral equations}
\label{sec:fredholm}

In this section, we consider Fredholm integral equations with nonnegative
kernels.  The purpose is not to introduce a new population model, but to
construct a determinant-free representation of the dominant eigensystem for
the integral operator induced by the kernel.  The reference point fixes the
normalization of the eigenfunctions at the kernel level, while the Bell
polynomials record the combinatorics generated by the resulting recursive
kernel expansion.

Throughout this section, we proceed in four steps.  First, we derive the
reference-point recursion on the admissible Hilbert--Schmidt subclass.  Next,
we show that the recursion is purely algebraic and therefore extends to
kernels in \(\mathbb K\), while its convergent representation at the dominant
spectral value remains conditional on the assumptions stated below.  We then
introduce two explicit assumptions governing spectral separation and reference
non-resonance.  Finally, under these assumptions we establish the absolute
convergence of the \(\Gamma_n\)-series and derive a determinant-free
representation of the dominant eigensystem.

\paragraph{Notation.}
Throughout this section, the reference pair $(x_0,y_0)$ is fixed unless otherwise stated. When the dependence on the reference pair is essential, we write
\[
\Gamma_n(x,y;x_0,y_0),
\qquad
w(x,y;x_0,y_0),
\qquad
v(y,x;x_0,y_0).
\]
When no confusion is likely, we suppress $(x_0,y_0)$ and simply write
\[
\Gamma_n(x,y),\qquad w(x,y),\qquad v(y,x).
\]
In later sections, when the diagonal choice $x_0=y_0$ is used, we further abbreviate
\[
\Gamma_n(x,y_0):=\Gamma_n(x,y_0;y_0,y_0),
\qquad
\Gamma_n^*(y_0,x):=\Gamma_n^*(y_0,y_0;y_0,x).
\]
This convention is used only to simplify notation; the underlying constructions remain those of the general reference-point framework.

\subsection{Structure of eigenfunction for the Hilbert--Schmidt case}
\label{ss:kernel-space}

We first investigate the structure of eigenfunctions for a Fredholm integral
operator with a positive kernel.

Let $\Omega^d \subseteq \mathbb{R}^d$ be a measurable domain and write $\mu$ for the Lebesgue measure on $\Omega^d$. The domain need not be bounded unless this is required by a separate compactness or Hilbert--Schmidt assumption. Let $K(x,y)$ be a measurable kernel on $\Omega^d \times \Omega^d$. To specify the class of kernels, we define the mixed-norm space by
\begin{equation}\label{eq:def-X}
\mathcal{X}
:=
L^\infty\!\bigl(\Omega^d;L^1(\Omega^d)\bigr),
\qquad
\|F\|_{\mathcal X}
:=
\operatorname*{ess\,sup}_{y\in\Omega^d}
\int_{\Omega^d}|F(x,y)|\,dx.
\end{equation}
Thus $F\in\mathcal X$ means that $x\mapsto F(x,y)$ is integrable for a.e.\ $y\in\Omega^d$, with an $L^1$ bound uniform in $y$.

Fix a reference pair \((x_0,y_0)\in\Omega^d\times\Omega^d\) and a relatively
compact open neighbourhood \(U\) of this pair.  Throughout this paper, we
denote by \(\mathbb K\) the class of kernels
\(K:\Omega^d\times\Omega^d\to\mathbb R\) such that:
\begin{itemize}
\item \textbf{Positivity:} $K(x,y)>0$ for all $x,y\in\Omega^d$.
\item \textbf{Mixed-norm integrability:} $K\in\mathcal X$.
\item \textbf{Local continuity at the reference pair:} $K$ is continuous on \(\overline U\).
\item \textbf{Uniform $L^\infty$-boundedness in the second variable:} there exists $M>0$ such that
\begin{equation}\label{eq:uniform-Linfty-second}
\operatorname*{ess\,sup}_{x\in\Omega^d}\|K(x,\cdot)\|_{L^\infty(\Omega^d)}\le M.
\end{equation}
\end{itemize}

\begin{remark}
\rm
No boundary condition is imposed on functions in \(\mathcal X\); in particular,
they are not assumed to vanish at the boundary of \(\Omega^d\).  The local
continuity assumption is used only to make point evaluation at the chosen
reference pair meaningful.  The neighbourhood \(U\) is fixed together with
the reference pair; it is not an additional state constraint.  The reference
pair is a normalization device for the eigensystem.  In applications it may
be chosen as a focal trait or stage, but the construction below does not
require this biological interpretation.
\end{remark}

\begin{remark}[On positivity]
\rm
Strict positivity is stronger than the irreducibility assumptions often used
in Perron--Frobenius theory.  We use it here to keep the presentation of the
reference-point formula uncluttered, not as the substantive spectral input of
the main theorem.  The latter is the dominant spectral separation assumption
below.  The same algebraic construction remains available under weaker
positivity-improving or irreducibility hypotheses, provided that spectral
separation and the reference non-resonance condition can be verified for the
induced operator.
\end{remark}

Let \(\mathbb K_2\subset\mathbb K\) denote the admissible
Hilbert--Schmidt subclass consisting of kernels that satisfy
\begin{equation}\label{eq:HS-assumption}
\int_{\Omega^d}\int_{\Omega^d}|K(x,y)|^2\,dx\,dy<\infty,
\qquad
\int_{\Omega^d}|K(x,x)|^2\,dx<\infty,
\end{equation}
and for which the classical Fredholm determinant \(D(\lambda)\), the first
Fredholm minor \(D(x,y;\lambda)\), and the corresponding resolvent formula
used below are well defined in the sense of classical kernel Fredholm theory
\cite{fredholm1903class}.

\noindent\textit{Role of the admissible Hilbert--Schmidt subclass.}
The class \(\mathbb K_2\) is used only to connect the reference-point
recursion with the classical determinant-and-minor construction.  We do not
assert that Hilbert--Schmidt membership alone makes the unregularized
determinant used below available for every abstract Hilbert--Schmidt operator.
The later algebraic recursion and its conditional extension are formulated on
\(\mathbb K\) and do not use a determinant.

Define
\[
\mathcal X_\star:=\{F\in\mathcal X:\ F \text{ is continuous on } \overline U\},
\qquad
\|F\|_\star:=\|F\|_{\mathcal X}+\sup_{(x,y)\in\overline U}|F(x,y)|.
\]
Then the point evaluation \(F\mapsto F(x_0,y_0)\) is a bounded functional on
\(\mathcal X_\star\).

\paragraph{Operators induced by kernels.}

Given $K\in\mathbb K$, define the operator $\mathbf K$ acting on the first variable by
\begin{equation}\label{eq:def-Kop}
(\mathbf K f)(x):=\int_{\Omega^d}K(x,\xi)f(\xi)\,d\xi,
\qquad f\in L^1(\Omega^d).
\end{equation}
We also define the kernel lift by
\begin{equation}\label{eq:def-Klift}
(\mathbf K F)(x,y):=\int_{\Omega^d}K(x,\xi)F(\xi,y)\,d\xi,
\qquad F\in\mathcal X.
\end{equation}
Then $\mathbf K F\in\mathcal X$.  The lift records all iterated kernels
\(K^{(n)}\) while retaining the second variable as a parameter.  The
eigenvalue problem remains the one-variable problem for \(\mathbf K\) on
\(L^1(\Omega^d)\), whereas the recursion defining \(\Gamma_n\) acts on
two-variable kernels.
We define the iterated kernels by
\begin{equation}\label{eq:def-iterated-kernels}
\begin{aligned}
K^{(1)}(x,y)&:=K(x,y),\\
K^{(n+1)}(x,y)&:=\int_{\Omega^d}K(x,\xi)K^{(n)}(\xi,y)\,d\xi,
\qquad n\ge1.
\end{aligned}
\end{equation}
Equivalently, $K^{(n)}(x,y)=(\mathbf K^{n-1}K)(x,y)$ for $n\ge1$.

On the algebraic class generated by the iterated kernels, with the pointwise
representatives fixed at the reference pair, define the rank-one correction
operation by
\begin{equation}\label{eq:def-P}
(\mathbf P F)(x,y):=K(x,y)F(x_0,y_0),
\end{equation}
and write $\mathbf A:=\mathbf K-\mathbf P$ for the corresponding taboo-type
operation:
\begin{equation}\label{eq:def-A}
(\mathbf A F)(x,y)
=
\int_{\Omega^d}K(x,\xi)F(\xi,y)\,d\xi
-
K(x,y)F(x_0,y_0).
\end{equation}

\begin{remark}[Algebraic role of the reference-point operator]
\label{rem:P_algebraic_role}
\rm
The operator $\mathbf P$ is introduced purely as an algebraic device on the kernel space. It does not represent a literal taboo event in the measure-theoretic sense, since a single point has Lebesgue measure zero in a continuous state space. Its role is to extract a rank-one component at the reference pair and thereby generate the recursion that underlies the determinant-free expansion.
The introduction of the reference pair is also natural from the viewpoint of eigenfunction normalization. Since an eigenfunction is determined only up to a nonzero multiplicative constant, one must fix its scale by prescribing its value at some reference point. The operator \(\mathbf P\) incorporates this normalization directly at the kernel level and thereby generates the recursive representation.
\end{remark}

\begin{remark}[Why ``taboo-type'']
\label{rem:why_taboo_type}
\rm
At this stage, the terminology ``taboo'' is only suggestive: the iterates $\{\Gamma_n\}_{n\ge1}$ are defined algebraically by the recursion $\Gamma_{n+1}=\mathbf A\Gamma_n$ and need not yet be interpreted as probabilities. The reason for the name is that, in the discrete-time matrix/Markov-chain setting discussed later, analogous quantities arise from subtracting a rank-one term that reinjects mass through a distinguished state. Here $\mathbf P$ plays exactly this algebraic role at the kernel level.
\end{remark}

\paragraph{Classical Fredholm-minor expansion and its coefficients.}

We consider the characteristic equation
\begin{equation}\label{eq:eigen-eq-K}
(\mathbf K w)(x,y)=\lambda_0 w(x,y),
\end{equation}
where $\lambda_0$ is the dominant positive eigenvalue and \(w\in\mathcal X\)
is a corresponding eigenfunction.  Such an eigenfunction is, of course,
determined only up to a nonzero scalar; the constants appearing below encode
the chosen normalization.  For general \(\lambda\) in the resolvent region,
the quantity introduced below is a normalized first-minor family; it becomes
an eigenfunction only when \(\lambda\) is specialized to a simple
characteristic value.  We first identify its coefficient expansion in the
admissible Hilbert--Schmidt case.

\begin{proposition}\label{prop:HS_eigenfunction}
Let $K\in\mathbb K_2$.  For \(|\lambda|\) sufficiently large in the domain
of the classical Fredholm resolvent expansion, the normalized first-minor
solution has the convergent expansion
\begin{equation}\label{eq:HS_structure}
w(x,y,x_0,y_0;\lambda)
=
c_0(x_0,y_0)\sum_{n=1}^\infty\frac{1}{\lambda^n}\Gamma_n(x,y,x_0,y_0),
\end{equation}
for a fixed reference point $(x_0,y_0)\in U$, where $\Gamma_1(x,y,x_0,y_0)=K^{(1)}(x,y)$ and, for $n\ge2$,
\begin{equation}\label{eq:HS_gamma}
\Gamma_n(x,y)
=
K^{(n)}(x,y)
+
\sum_{\ell=1}^{n-1}(-1)^\ell\sum_{k=\ell}^{n-1}
K^{(n-k)}(x,y)\,
\widehat B_{k,\ell}\!\left(K^{(1)},K^{(2)},\dots,K^{(k)}\right),
\end{equation}
where \(\widehat B\) denotes the ordinary partial Bell polynomial, defined by
\[
\left(\sum_{j\ge1} z_j t^j\right)^\ell
=
\sum_{m\ge0}\widehat B_{m,\ell}(z_1,z_2,\ldots)t^m .
\]
Equivalently, for \(\ell\ge1\),
\[
\widehat B_{m,\ell}(z_1,z_2,\ldots)
:=
\sum_{\substack{j_1+\cdots+j_\ell=m\\ j_i\ge1}}
z_{j_1}\cdots z_{j_\ell},
\qquad
\widehat B_{0,0}:=1 .
\]
Thus ``ordinary'' refers to ordinary generating functions; no factorial
weights are included \cite{riordan1958,comtet1974}.
Pointwise absolute convergence of \eqref{eq:HS_structure} at
\(\lambda=\lambda_0\) is established only by
Corollary~\ref{cor:gamma_normalized_convergence} under
Assumptions~\ref{ass:dominant_spectral_separation} and
\ref{ass:reference_generating_nonresonance}.
\end{proposition}

\begin{proof}[Derivation from Fredholm theory]
For \(K\in\mathbb K_2\), the Fredholm determinant \(D(\lambda)\) and the
first Fredholm minor are well defined.  For \(|\lambda|\) sufficiently large,
the resolvent expansion gives
\[
\frac{1}{\lambda}
\left(\left(\mathbf I-\frac{1}{\lambda}\mathbf K\right)^{-1}K\right)(x,y)
=
\sum_{n\ge1}\lambda^{-n}K^{(n)}(x,y).
\]
Let
\[
D(x,y;\lambda)
:=
D(\lambda)
\left(\mathbf I-\frac{1}{\lambda}\mathbf K\right)^{-1}
\frac{K(x,y)}{\lambda}.
\]
If \(\lambda_0\) is a simple zero of \(D\), then
\(D(\cdot,\cdot;\lambda_0)\) is a nonzero eigenfunction for the eigenvalue
\(\lambda_0\).  Choose \((x_0,y_0)\) so that
\(D(x_0,y_0;\lambda_0)\ne0\), and normalize
\[
w(x,y;\lambda)=c(x_0,y_0;\lambda)D(x,y;\lambda),\qquad
c(x_0,y_0;\lambda)=
\frac{c_0}{D(\lambda)+D(x_0,y_0;\lambda)} .
\]
Then
\[
\frac{c(x_0,y_0;\lambda)}{c_0}D(\lambda)
=1-\frac{w(x_0,y_0;\lambda)}{c_0}.
\]
Substitution into the Fredholm resolvent equation yields
\[
w(x,y;\lambda)
=
\frac{K(x,y)}{\lambda}\bigl(c_0-w(x_0,y_0;\lambda)\bigr)
+\frac1\lambda(\mathbf K w)(x,y;\lambda).
\]
Writing \(w(x,y;\lambda)=\sum_{n\ge1}c_0\lambda^{-n}\Gamma_n(x,y)\) and
comparing powers of \(\lambda^{-1}\) gives
\[
\Gamma_1=K,\qquad
\Gamma_{n+1}=\mathbf K\Gamma_n-K\,\Gamma_n(x_0,y_0).
\]
Thus \(\Gamma_n=(\mathbf K-\mathbf P)^{n-1}K\).  Expanding the words in
\(\mathbf K\) and the rank-one operator \(\mathbf P\) gives the displayed
Bell-polynomial formula; the detailed word expansion is recorded below in
Lemma~\ref{lem:Gamma_Bell_nonHS}.  Evaluating the normalized Fredholm minor at
\(\lambda=\lambda_0\) identifies the associated eigenfunction when
\(\lambda_0\) is a simple zero of \(D\).
\end{proof}

The spectral assumptions introduced below are \emph{additional assumptions}
and do not follow from $K\in\mathbb K$ alone.

\paragraph{Construction of the kernel recursion.}

The recursion derived below from the classical Fredholm expansion on
\(\mathbb K_2\) is algebraic and is defined for every \(K\in\mathbb K\).

\smallskip
\noindent\textbf{Step 1. Fredholm-side derivation of the recursion for $K\in\mathbb K_2$.}
For $K\in\mathbb K_2$, the classical Fredholm determinant and first minor are
well defined by the admissibility built into the definition of \(\mathbb K_2\),
and, for large \(|\lambda|\), the resolvent of $\mathbf K$ is given by the
Neumann series \cite{fredholm1903class,neumann1877untersuchungen}
\begin{equation}\label{eq:neumann_resolvent_K}
\frac{1}{\lambda}\left(\left(\mathbf I-\frac{1}{\lambda}\mathbf K\right)^{-1}K\right)(x,y)
=
\sum_{n=1}^\infty\frac{1}{\lambda^n}K^{(n)}(x,y),
\qquad |\lambda|>\|\mathbf K\|_{\mathrm{op}}.
\end{equation}
We define
\begin{equation}\label{eq:def_Dxy_lambda}
D(x,y;\lambda)
:=
D(\lambda)\left(\mathbf I-\frac{1}{\lambda}\mathbf K\right)^{-1}\frac{K(x,y)}{\lambda}.
\end{equation}
If $\lambda_0$ is a zero of $D(\lambda)$ corresponding to the dominant simple eigenvalue, then $D(x,y;\lambda_0)$ is a nontrivial eigenfunction of $\mathbf K$. Choose a reference pair $(x_0,y_0)$ so that
\begin{equation}\label{eq:ref_pair_nonzero_D}
D(x_0,y_0;\lambda_0)\neq0.
\end{equation}
We normalize by
\begin{equation}\label{eq:def_w_normalized}
w(x,y;\lambda)
=
w(x,y;x_0,y_0;\lambda)
:=
c(x_0,y_0;\lambda)\,D(x,y;\lambda),
\end{equation}
with
\begin{equation}\label{eq:def_c_normalized}
c(x_0,y_0;\lambda)
:=
\frac{c_0}{D(\lambda)+D(x_0,y_0;\lambda)},
\qquad c_0=c_0(x_0,y_0)\neq0.
\end{equation}
Combining \eqref{eq:def_w_normalized} and \eqref{eq:def_c_normalized}, we obtain
\begin{equation}\label{eq:key_relation_cD}
\frac{c(x_0,y_0;\lambda)}{c_0}D(\lambda)
=
1-\frac{w(x_0,y_0;\lambda)}{c_0}.
\end{equation}

Substituting \eqref{eq:def_Dxy_lambda} into \eqref{eq:def_w_normalized}, we obtain
\[
\left(\mathbf I-\frac{1}{\lambda}\mathbf K\right)w(x,y;\lambda)
=
c(x_0,y_0;\lambda)D(\lambda)\frac{K(x,y)}{\lambda}.
\]
Hence
\begin{equation}\label{eq:fredholm_rewritten}
\begin{aligned}
w(x,y;\lambda)
&=
c(x_0,y_0;\lambda)D(\lambda)\frac{K(x,y)}{\lambda}
+\frac{1}{\lambda}\mathbf K w(x,y;\lambda)\\
&=
c(x_0,y_0;\lambda)D(\lambda)\frac{K(x,y)}{\lambda}
+
\int_{\Omega^d}\frac{K(x,\xi)}{\lambda}w(\xi,y;\lambda)\,d\xi.
\end{aligned}
\end{equation}
Set formally
\begin{equation}\label{eq:w_Gamma_series_formal}
w(x,y;\lambda)
=
\sum_{n=1}^\infty\frac{c_0}{\lambda^n}\Gamma_n(x,y).
\end{equation}
Using \eqref{eq:key_relation_cD}, equation \eqref{eq:fredholm_rewritten} becomes
\begin{equation}\label{eq:fredholm_rewritten_2}
w(x,y;\lambda)
=
\frac{K(x,y)}{\lambda}\bigl(c_0-w(x_0,y_0;\lambda)\bigr)
+
\frac{1}{\lambda}(\mathbf K w)(x,y;\lambda).
\end{equation}
Substituting \eqref{eq:w_Gamma_series_formal} into both sides of \eqref{eq:fredholm_rewritten_2} and comparing coefficients of $\lambda^{-n}$, we obtain the recursion
\begin{equation}\label{eq:Gamma_recursion_main}
\begin{aligned}
\Gamma_1(x,y)&=K(x,y),\\
\Gamma_{n+1}(x,y)
&=
(\mathbf K\Gamma_n)(x,y)-K(x,y)\Gamma_n(x_0,y_0),
\qquad n\ge1.
\end{aligned}
\end{equation}

\smallskip
\noindent\textbf{Step 2. Operator form of the recursion.}
Recall the rank-one reference-point operator
\begin{equation}\label{eq:def_P_refpoint_again}
(\mathbf P F)(x,y):=K(x,y)F(x_0,y_0),
\end{equation}
and set
\begin{equation}\label{eq:def_A_refpoint_again}
\mathbf A:=\mathbf K-\mathbf P.
\end{equation}
Then \eqref{eq:Gamma_recursion_main} is simply
\begin{equation}\label{eq:Gamma_recursion_operator}
\Gamma_1=K,
\qquad
\Gamma_{n+1}=\mathbf A\Gamma_n,
\qquad n\ge1,
\end{equation}
that is,
\begin{equation}\label{eq:Gamma_A_power}
\Gamma_n=\mathbf A^{\,n-1}K,
\qquad n\ge1.
\end{equation}

\begin{lemma}[Bell-polynomial representation of $\Gamma_n$]
\label{lem:Gamma_Bell_nonHS}
Let $K\in\mathbb K$. Then, for every $n\ge1$,
\begin{equation}\label{eq:Gamma_Bell_nonHS}
\Gamma_n
=
\sum_{k=1}^{n}K^{(k)}
\sum_{\ell=0}^{n-k}(-1)^\ell\,
\widehat B_{n-k,\ell}(b_1,b_2,\dots),
\end{equation}
and consequently
\begin{equation}\label{eq:a_Bell_nonHS}
\Gamma_n(x_0,y_0)
=
\sum_{k=1}^{n}b_k
\sum_{\ell=0}^{n-k}(-1)^\ell\,
\widehat B_{n-k,\ell}(b_1,b_2,\dots,b_{n-k}),
\end{equation}
where
\[
b_n:=K^{(n)}(x_0,y_0),
\qquad n\ge1.
\]
\end{lemma}

\begin{proof}
\noindent\textbf{Step 1.}
Expanding $(\mathbf K-\mathbf P)^{n-1}$ yields
\[
(\mathbf K-\mathbf P)^{n-1}
=
\sum_{\ell=0}^{n-1}(-1)^\ell
\sum_{\substack{i_0,\dots,i_\ell\ge0\\ i_0+\cdots+i_\ell=n-1-\ell}}
\mathbf K^{i_0}\mathbf P\,\mathbf K^{i_1}\mathbf P\cdots\mathbf P\,\mathbf K^{i_\ell}.
\]
Applying this to $K$ gives
\begin{equation}\label{eq:Gamma_word_expansion}
\Gamma_n
=
\sum_{\ell=0}^{n-1}(-1)^\ell
\sum_{\substack{i_0,\dots,i_\ell\ge0\\ i_0+\cdots+i_\ell=n-1-\ell}}
\mathbf K^{i_0}\mathbf P\,\mathbf K^{i_1}\mathbf P\cdots\mathbf P\,\mathbf K^{i_\ell}K.
\end{equation}

\smallskip
\noindent\textbf{Step 2.}
Since $\mathbf K^jK=K^{(j+1)}$ for $j\ge0$, we have
\[
\mathbf P\,\mathbf K^jK
=
K\,(\mathbf K^jK)(x_0,y_0)
=
K\,K^{(j+1)}(x_0,y_0)
=
K\,b_{j+1}.
\]
Thus each word in \eqref{eq:Gamma_word_expansion} collapses to
\[
\mathbf K^{i_0}\mathbf P\,\mathbf K^{i_1}\mathbf P\cdots\mathbf P\,\mathbf K^{i_\ell}K
=
K^{(i_0+1)}\prod_{r=1}^{\ell}b_{i_r+1}.
\]
Therefore
\begin{equation}\label{eq:Gamma_comp_expansion}
\Gamma_n
=
\sum_{\ell=0}^{n-1}(-1)^\ell
\sum_{\substack{i_0,\dots,i_\ell\ge0\\ i_0+\cdots+i_\ell=n-1-\ell}}
K^{(i_0+1)}\prod_{r=1}^{\ell}b_{i_r+1}.
\end{equation}

\smallskip
\noindent\textbf{Step 3.}
Fix $n\ge1$ and set $k:=i_0+1\in\{1,\dots,n\}$. Then $m:=n-k$ and
\[
i_1+\cdots+i_\ell=m-\ell.
\]
Writing $j_r:=i_r+1\ge1$ for $1\le r\le\ell$, we obtain
\[
j_1+\cdots+j_\ell=m,
\qquad
\prod_{r=1}^{\ell}b_{i_r+1}=\prod_{r=1}^{\ell}b_{j_r}.
\]
Hence
\[
\sum_{\substack{i_1,\dots,i_\ell\ge0\\ i_1+\cdots+i_\ell=m-\ell}}
\prod_{r=1}^{\ell}b_{i_r+1}
=
\sum_{\substack{j_1,\dots,j_\ell\ge1\\ j_1+\cdots+j_\ell=m}}
\prod_{r=1}^{\ell}b_{j_r}
=
\widehat B_{m,\ell}(b_1,b_2,\dots).
\]
Substituting this into \eqref{eq:Gamma_comp_expansion} yields \eqref{eq:Gamma_Bell_nonHS}, and evaluating at $(x_0,y_0)$ gives \eqref{eq:a_Bell_nonHS}.
\end{proof}

\paragraph{Spectral assumptions for the non-Hilbert--Schmidt extension.}

The class $\mathbb K$ ensures that the reference-point construction is well posed at the kernel level, but it does \emph{not} by itself imply quasi-compactness, simplicity of the spectral radius, or the existence of a Riesz decomposition. For the extension beyond $\mathbb K_2$, we therefore impose the following spectral assumption explicitly.

\begin{assumption}[Dominant spectral separation]
\label{ass:dominant_spectral_separation}
Let $\mathbf K:L^1(\Omega^d)\to L^1(\Omega^d)$ be the integral operator induced by a kernel $K\in\mathbb K$. Assume that:
\begin{enumerate}
\item[(i)] $\lambda_0:=\rho(\mathbf K)>0$ is an isolated algebraically simple eigenvalue of $\mathbf K$;
\item[(ii)] there exist a rank-one Riesz projection with kernel
\(\Pi(x,y)\), a number $\theta\in(0,\lambda_0)$, and a function
$C_1(x,y)>0$ such that
\begin{equation}\label{eq:K_n_riesz_kernel_form}
K^{(n)}(x,y)
=
\lambda_0^{\,n}\Pi(x,y)+R_n(x,y),
\qquad
|R_n(x,y)|\le C_1(x,y)\theta^n,
\qquad n\ge1;
\end{equation}
\item[(iii)] the reference pair $(x_0,y_0)$ is chosen so that
\(\kappa:=\Pi(x_0,y_0)>0\).  We then define the normalized product
\begin{equation}\label{eq:def-U-riesz}
U(x,y):=\frac{\Pi(x,y)}{\kappa}=u(x)v(y),
\qquad u(x_0)=v(y_0)=1,
\end{equation}
so that
\begin{equation}\label{eq:U_normalization_refpoint}
U(x_0,y_0)=1,
\qquad
\Pi(x,y)=\kappa U(x,y).
\end{equation}
\end{enumerate}
\end{assumption}

\begin{remark}
\rm
Assumption~\ref{ass:dominant_spectral_separation} is not a consequence of the definition of $\mathbb K$; it is an additional spectral hypothesis used to isolate the Perron contribution at the kernel level. In particular, the pointwise decomposition \eqref{eq:K_n_riesz_kernel_form} is the precise place where the dominant rank-one contribution and the exponentially smaller remainder are separated. The pointwise bound is stronger than an operator-norm Riesz decomposition; it is imposed here because the reference-point construction and the resulting demographic quantities are evaluated at fixed kernel points.  The normalization in item (iii) does not restrict the kernel or the operator: the positive right and left eigenfunctions may be scaled so that $u(x_0)=v(y_0)=1$.  The invariant Riesz projection is $\Pi=\kappa U$; the scalar $\kappa$ records the normalization required to distinguish the projection kernel from the normalized eigenfunction product.
\end{remark}

\paragraph{Sufficient conditions for Assumption~\ref{ass:dominant_spectral_separation}.}

We record standard mechanisms by which Assumption~\ref{ass:dominant_spectral_separation}
can be verified.  The following remarks are intended as checkable routes, not
as substitutes for verification in a concrete model.  They are based on
classical Perron--Frobenius, Krein--Rutman, and quasi-compactness theory
\cite{kreinrutman1948,henry1981,hennion1993}.  The additional reference
non-resonance condition below is scalar and depends on the chosen reference
pair; it is therefore stated separately from these operator-theoretic
conditions.

\begin{remark}[Hilbert--Schmidt case]
\label{rem:HS_sufficient}
\rm
If $K\in\mathbb K_2$ is a positive Hilbert--Schmidt kernel, then the induced operator $\mathbf K$ on $L^2(\Omega^d)$ is compact. 
Under the usual irreducibility or positivity-improving conditions, classical results imply that $\rho(\mathbf K)$ is a simple positive eigenvalue and that the iterates admit a decomposition of the form
\[
K^{(n)}(x,y)
=
\rho(\mathbf K)^n u(x)v(y)+R_n(x,y),
\]
with exponential decay of $R_n$ in operator norm.  Assumption~\ref{ass:dominant_spectral_separation} follows when the kernel has sufficient additional regularity to upgrade this to the pointwise bound in \eqref{eq:K_n_riesz_kernel_form}.  Thus Hilbert--Schmidt membership alone is not asserted to imply the pointwise assumption.
\end{remark}

\begin{remark}[Doeblin-type condition]
\label{rem:doeblin}
\rm
If, in addition to the boundedness assumptions above, the operator satisfies a
Doeblin-type minorization: there exist measurable sets
\(A,B\subset\Omega^d\) with positive measure and \(\delta>0\) such that
\[
K(x,y)\ge \delta\,\mathbf 1_A(x)\mathbf 1_B(y),
\]
and the complementary part is controlled in a standard Lasota--Yorke or
compactness estimate, then \(\mathbf K\) is quasi-compact on \(L^1(\Omega^d)\).
Standard quasi-compactness theory yields a spectral gap in operator norm.  If
the kernel has sufficient regularity to pass this estimate to pointwise
representatives, the iterates have a decomposition of the form
\[
K^{(n)}(x,y)
=
\rho(\mathbf K)^n u(x)v(y)+R_n(x,y),
\]
with the pointwise exponential decay required in
\eqref{eq:K_n_riesz_kernel_form}. Under these additional hypotheses,
Assumption~\ref{ass:dominant_spectral_separation} holds.
\end{remark}

\begin{remark}[Rank-one perturbation]
\label{rem:rank_one}
\rm
If the induced operator admits a compatible spectral decomposition
\[
\mathbf K=\lambda_0\mathbf U+\mathbf R,
\]
where \(\mathbf U\) is the rank-one spectral projection with kernel
\(u(x)v(y)\), \(\mathbf U\mathbf R=\mathbf R\mathbf U=0\), and
\(\rho(\mathbf R)<\lambda_0\), then the operator iterates have the corresponding
spectral decomposition.  If, in addition, the kernels of \(\mathbf R^n\) obey
the required pointwise exponential bound, the iterated kernels satisfy
\[
K^{(n)}(x,y)
=
\lambda_0^n u(x)v(y)+R_n(x,y),
\]
with exponential decay governed by any rate strictly between
$\rho(\mathbf R)$ and $\lambda_0$. Hence
Assumption~\ref{ass:dominant_spectral_separation} holds under these compatibility
and pointwise-regularity conditions.
\end{remark}

\begin{example}[Convolution-type kernel]
\label{ex:convolution}
Let
\[
K(x,y)=\phi(x-y)g(y),
\]
where $\phi\ge0$ is continuous with $\phi(0)>0$ and $g(y)>0$ is bounded. 
Then $K$ provides a local Doeblin-type minorization near the diagonal.  This is
one ingredient of Remark~\ref{rem:doeblin}; Assumption
\ref{ass:dominant_spectral_separation} follows when the complementary
quasi-compactness and pointwise-regularity requirements stated there are also
verified.
\end{example}

\begin{lemma}[Reference-point generating identity]
\label{lem:reference_generating_identity}
For \(|t|\) sufficiently small one has
\begin{equation}\label{eq:reference_generating_identity}
\Gamma(t;x,y)
:=
\sum_{n\ge1}\Gamma_n(x,y)t^n
=
\frac{\mathcal K(t;x,y)}
{1+\mathcal K(t;x_0,y_0)},
\qquad
\mathcal K(t;x,y):=\sum_{n\ge1}K^{(n)}(x,y)t^n .
\end{equation}
\end{lemma}

\begin{proof}
The identity follows directly from the word expansion
\eqref{eq:Gamma_comp_expansion}.  A word containing \(\ell\) occurrences
of the reference-point correction contributes
\[
(-1)^\ell K^{(i_0+1)}(x,y)\prod_{r=1}^{\ell}b_{i_r+1}
\]
to \(\Gamma_n(x,y)\), where
\[
n=(i_0+1)+\sum_{r=1}^{\ell}(i_r+1).
\]
After multiplication by \(t^n\) and summation over all
\(i_0,\ldots,i_\ell\ge0\), this contribution becomes
\[
\mathcal K(t;x,y)\bigl(-\mathcal K(t;x_0,y_0)\bigr)^\ell .
\]
Summing over \(\ell\ge0\) gives
\[
\Gamma(t;x,y)
=
\mathcal K(t;x,y)
\sum_{\ell\ge0}\bigl(-\mathcal K(t;x_0,y_0)\bigr)^\ell
=
\frac{\mathcal K(t;x,y)}
{1+\mathcal K(t;x_0,y_0)}
\]
where the geometric series is convergent.  The identity also holds formally,
and hence analytically for sufficiently small \(|t|\).
\end{proof}

\begin{assumption}[Reference generating non-resonance]
\label{ass:reference_generating_nonresonance}
Put
\[
r_n:=R_n(x_0,y_0),
\qquad
E(t):=\sum_{n\ge1}r_nt^n .
\]
There exists a number \(r_\ast\) such that
\[
\lambda_0^{-1}<r_\ast<\theta^{-1}
\]
and
\begin{equation}\label{eq:reference_denominator_nonzero}
1+(\kappa-1)\lambda_0t+(1-\lambda_0t)E(t)\ne0
\qquad (|t|\le r_\ast).
\end{equation}
\end{assumption}

\begin{remark}
\rm
Assumption~\ref{ass:reference_generating_nonresonance} is a scalar
non-resonance condition for the chosen reference pair.  A direct sufficient
condition is
\[
|(\kappa-1)\lambda_0t+(1-\lambda_0t)E(t)|<1
\qquad (|t|\le r_\ast).
\]
Its
role is to prevent the reference denominator from producing singularities
before the circle on which the normalized \(\Gamma\)-series is evaluated.
Without such a condition, the Bell-polynomial coefficients may have
singularities generated by the scalar reference denominator, and
ordinary absolute convergence at \(t=\lambda_0^{-1}\) does not follow
from the Perron--remainder decomposition alone.  This condition does not
change membership in the ambient kernel class \(\mathbb K\), nor does it impose
compactness in an operator norm; it restricts the choice of reference pair for
the ordinary convergent-series representation.  Hilbert--Schmidt membership by
itself does not imply this scalar condition.
\end{remark}

\begin{remark}[Logical roles of the two assumptions]
\label{rem:logical_roles_assumptions}
\rm
Assumptions~\ref{ass:dominant_spectral_separation} and
\ref{ass:reference_generating_nonresonance} are kept separate because they are
logically distinct.  The first is a condition on the induced operator and its
iterated kernel; the second is a scalar condition attached to the chosen
reference pair.  The recursion, the Bell-polynomial identity, and the formal
generating identity preceding this remark require neither assumption.  Both
assumptions are invoked only when ordinary absolute convergence at
\(t=\lambda_0^{-1}\), and conclusions based on that convergence, are asserted.
\end{remark}

\begin{remark}[Reference condition under an exact Perron decomposition]
\label{rem:reference_condition_automatic}
\rm
If the Perron decomposition is exact at the kernel level, that is,
\[
K^{(n)}(x,y)=\kappa\lambda_0^n U(x,y)
\qquad (n\ge1),
\]
then \(R_n\equiv0\), hence \(E(t)\equiv0\).  The reference denominator is
\(1+(\kappa-1)\lambda_0t\).  The condition holds precisely when this linear
factor has no zero on the chosen disk.
Small perturbations also satisfy the condition whenever the full denominator
in \eqref{eq:reference_denominator_nonzero} remains nonzero there.
\end{remark}

\begin{corollary}[Absolute convergence of the normalized $\Gamma$-series at $\lambda_0$]
\label{cor:gamma_normalized_convergence}
Assume Assumptions~\ref{ass:dominant_spectral_separation} and
\ref{ass:reference_generating_nonresonance}.  Then, for every fixed
\((x,y)\) and every
\(r\) with
\[
\lambda_0^{-1}<r<r_\ast,
\]
there exists a function \(C_r(x,y)>0\) such that
\begin{equation}\label{eq:normalized_gamma_bound}
\left|
\lambda_0^{-n}\Gamma_n(x,y)
\right|
\le
C_r(x,y)(\lambda_0 r)^{-n},
\qquad
n\ge1.
\end{equation}
In particular,
\begin{equation}\label{eq:normalized_gamma_series}
\sum_{n=1}^{\infty}\lambda_0^{-n}\Gamma_n(x,y)
\end{equation}
converges absolutely.
\end{corollary}

\begin{proof}
Fix \(x,y\) and \(r\) with \(\lambda_0^{-1}<r<r_\ast\).  By
Assumption~\ref{ass:dominant_spectral_separation}, the series
\[
R(t;x,y):=\sum_{n\ge1}R_n(x,y)t^n,
\qquad
E(t):=\sum_{n\ge1}R_n(x_0,y_0)t^n
\]
are analytic for \(|t|<\theta^{-1}\).  The identity in
Lemma~\ref{lem:reference_generating_identity}, initially obtained as a formal
power-series identity and hence as an analytic identity near the origin,
therefore continues under Assumptions
\ref{ass:dominant_spectral_separation} and
\ref{ass:reference_generating_nonresonance} to the analytic representation
\[
\Gamma(t;x,y)
=
\frac{\kappa\lambda_0t\,U(x,y)+(1-\lambda_0t)R(t;x,y)}
{1+(\kappa-1)\lambda_0t+(1-\lambda_0t)E(t)}
\]
on a neighbourhood of the closed disk \(|t|\le r\).  At
\(t=\lambda_0^{-1}\), the denominator equals \(\kappa>0\) and the numerator
equals \(\kappa U(x,y)\); hence the singularity is removable.  Thus
\[
M_r(x,y):=\sup_{|t|=r}|\Gamma(t;x,y)|<\infty .
\]
Cauchy's coefficient estimate gives
\[
|\Gamma_n(x,y)|\le M_r(x,y)r^{-n}.
\]
Therefore
\[
\left|\lambda_0^{-n}\Gamma_n(x,y)\right|
\le M_r(x,y)(\lambda_0 r)^{-n}.
\]
This is \eqref{eq:normalized_gamma_bound} with
\(C_r(x,y)=M_r(x,y)\).  Since \(\lambda_0 r>1\), the resulting
majorant is summable.
The conclusion is pointwise in \((x,y)\); no uniform convergence over
\(\Omega^d\times\Omega^d\) is asserted without an additional uniform bound on
\(M_r(x,y)\).
\end{proof}

Under Assumptions~\ref{ass:dominant_spectral_separation} and
\ref{ass:reference_generating_nonresonance},
Corollary~\ref{cor:gamma_normalized_convergence} implies that the normalized series
\eqref{eq:w_Gamma_series_formal} converges absolutely at $\lambda=\lambda_0$. Therefore the formal recursion \eqref{eq:Gamma_recursion_main} yields a well-defined eigenfunction representation at the dominant spectral value.

\begin{remark}[Role of the $\Gamma$-series beyond the Perron case]
\label{rem:peripheral_spectrum}
\rm
The series representation
\[
w(x,y)
=
\sum_{n\ge1}\frac{1}{\lambda_0^n}\Gamma_n(x,y)
\]
should be interpreted with care when the peripheral spectrum of $\mathbf K$ contains multiple eigenvalues of modulus $\lambda_0$.

In the Perron--Frobenius case, where $\lambda_0$ is a simple dominant eigenvalue and no other spectral values lie on the circle $|\lambda|=\lambda_0$, one has
\[
\lambda_0^{-n}K^{(n)}(x,y)\to \Pi(x,y)=\kappa U(x,y),
\]
and the $\Gamma$-series converges in the usual sense. However, in the presence of peripheral eigenvalues $\lambda_j=\lambda_0e^{i\theta_j}$, the iterates may satisfy
\[
\lambda_0^{-n}K^{(n)}(x,y)=\sum_j e^{in\theta_j}U_j(x,y)+o(1),
\]
so that the general term does not decay. As a consequence, the series
\[
\sum_{n\ge1}\lambda_0^{-n}\Gamma_n(x,y)
\]
typically fails to converge termwise.

Nevertheless, the resolvent quotient
\[
\frac{(\lambda I-\mathbf K)^{-1}K(x,y)}
{1+(\lambda I-\mathbf K)^{-1}K(x_0,y_0)}
\]
may still admit a finite limit as $\lambda\downarrow\lambda_0$ along the real axis. In that situation, the $\Gamma$-series should be understood in an Abel-type sense rather than as an ordinary convergent series at $\lambda=\lambda_0$.
\end{remark}

\paragraph{Adjoint counterpart.}

In applications, the adjoint eigenfunction plays the role of a reproductive
value.  Its reference-point expansion is as follows.

\begin{remark}[Adjoint expansion]
\label{rem:adjoint_expansion}
\rm
Let $\mathbf K^*$ act on functions $v(y,\cdot)\in L^\infty(\Omega^d)$ by
\begin{equation}\label{eq:def-Kstar}
(\mathbf K^*v)(y,x)
:=
\int_{\Omega^d}v(y,\xi)K(\xi,x)\,d\xi.
\end{equation}
Consider the adjoint eigen-equation
\begin{equation}\label{eq:adjoint_eigen_eq}
v(y,x;\lambda)
=
\frac{1}{\lambda}(\mathbf K^*v)(y,x;\lambda)
=
\frac{1}{\lambda}\int_{\Omega^d}v(y,\xi;\lambda)K(\xi,x)\,d\xi.
\end{equation}
Introduce the transpose kernel
\[
K^\top(x,y):=K(y,x),
\]
and define the adjoint reference-point operator $\mathbf P^*$ by
\begin{equation}\label{eq:def-P_star}
(\mathbf P^*G)(y,x):=G(y_0,x_0)K^\top(x,y).
\end{equation}
Then $\mathbf P^*$ is rank one. Define the adjoint taboo-type iterates by
\begin{equation}\label{eq:Gamma_star_recursion}
\Gamma_1^*:=K^\top,
\qquad
\Gamma_{n+1}^*:=(\mathbf K^*-\mathbf P^*)\Gamma_n^*,
\qquad n\ge1,
\end{equation}
equivalently,
\[
\Gamma_n^*=(\mathbf K^*-\mathbf P^*)^{n-1}K^\top.
\]
Keeping the reference variables explicit, the Bell-polynomial coefficient
formula gives
\begin{equation}\label{eq:Gamma_star_transpose_relation}
\Gamma_n^*(y,x;y_0,x_0)=\Gamma_n(y,x;y_0,x_0),
\end{equation}
where the right-hand side is the reference-point coefficient with the variables
and the reference pair interchanged.  Thus the adjoint coefficients use the
reference scalars \(K^{(n)}(y_0,x_0)\), not in general
\(K^{(n)}(x_0,y_0)\).
In kernel form,
\[
\Gamma_{n+1}^*(y,x)
=
\int_{\Omega^d}\Gamma_n^*(y,\xi)K(\xi,x)\,d\xi
-
\Gamma_n^*(y_0,x_0)K^\top(x,y).
\]
Consequently, for \(|\lambda|\) sufficiently large, the adjoint solution
admits the expansion
\begin{equation}\label{eq:adjoint_series_final}
v(y,x;\lambda)
=
\frac{c_1}{\lambda}
\left[\left(I-\frac{1}{\lambda}(\mathbf K^*-\mathbf P^*)\right)^{-1}K^\top\right](x,y)
=
c_1\sum_{n=1}^{\infty}\frac{1}{\lambda^n}\Gamma_n^*(y,x),
\end{equation}
with $c_1=c_1(x_0,y_0)\neq0$.  To evaluate this series at
\(\lambda=\lambda_0\), the interchanged reference pair must satisfy the scalar
condition corresponding to Assumption
\ref{ass:reference_generating_nonresonance}.  More precisely, put
\[
\kappa^*:=\Pi(y_0,x_0)>0,
\qquad
E^*(t):=\sum_{n\ge1}R_n(y_0,x_0)t^n,
\]
and assume that, for some
\(\lambda_0^{-1}<r_*^*<\theta^{-1}\),
\[
1+(\kappa^*-1)\lambda_0t
 +(1-\lambda_0t)E^*(t)\ne0
\qquad (|t|\le r_*^*).
\]
Then \eqref{eq:Gamma_star_transpose_relation} and
Corollary~\ref{cor:gamma_normalized_convergence}, applied to the interchanged
reference pair, imply pointwise absolute convergence of the adjoint series at
\(\lambda_0\).  In the diagonal-reference applications below,
\(x_0=y_0\), so this is exactly the original reference non-resonance condition
and no additional scalar assumption is required.
\end{remark}

\paragraph{Spectral analysis.}

\begin{proposition}\label{prop:dominant_eigenvalue_riesz}
Assume Assumptions~\ref{ass:dominant_spectral_separation} and
\ref{ass:reference_generating_nonresonance}. Then there exists a nonnegative
nontrivial eigenfunction $w_\rho$ such that
\[
\mathbf K w_\rho=\rho(\mathbf K)\,w_\rho.
\]
Moreover,
\begin{equation}\label{eq:w_rho_gamma_representation}
w_\rho(x,y)
=
c_0\sum_{n=1}^{\infty}\frac{\Gamma_n(x,y)}{\rho(\mathbf K)^n},
\end{equation}
where $c_0=w_\rho(x_0,y_0)$, and the normalization identity
\begin{equation}\label{eq:reference_identity_rho}
1=
\sum_{n=1}^{\infty}\frac{\Gamma_n(x_0,y_0)}{\rho(\mathbf K)^n}
\end{equation}
holds whenever $c_0\neq0$.
\end{proposition}

\begin{proof}
Since $\rho(\mathbf K)=\lambda_0$ is assumed to be a simple isolated eigenvalue, there exists a nonnegative nontrivial eigenfunction
\[
w_\rho\not\equiv0,
\qquad
\mathbf K w_\rho=\rho(\mathbf K)\,w_\rho.
\]
By Corollary~\ref{cor:gamma_normalized_convergence}, the series
\[
\sum_{n=1}^{\infty}\frac{\Gamma_n(x,y)}{\rho(\mathbf K)^n}
\]
converges absolutely because \(\rho(\mathbf K)=\lambda_0\).  Put
\[
S(x,y):=\sum_{n=1}^{\infty}\frac{\Gamma_n(x,y)}{\rho(\mathbf K)^n}.
\]
By Corollary~\ref{cor:gamma_normalized_convergence}, the analytic
representation is valid near \(t=\lambda_0^{-1}\) and hence may be evaluated
there:
\[
S(x,y)
=\Gamma(\lambda_0^{-1};x,y)
=U(x,y).
\]
Since \(U=\Pi/\kappa\) is a scalar normalization of the kernel of the rank-one
Riesz projection associated with \(\lambda_0\),
\[
\mathbf K U=\lambda_0U.
\]
Consequently \(S\) is an eigenfunction for \(\rho(\mathbf K)=\lambda_0\),
and \(U(x_0,y_0)=1\) gives \(S(x_0,y_0)=1\).  Since the eigenvalue is simple,
\(w_\rho=c_0S\) for a nonzero scalar \(c_0\), and
\eqref{eq:w_rho_gamma_representation} follows.

Finally, evaluating \(S(x_0,y_0)=1\) in its defining series gives
\eqref{eq:reference_identity_rho}.
\end{proof}

\subsection{Non-Hilbert--Schmidt solution and its property}

\begin{theorem}[Main theorem: non-Hilbert--Schmidt extension]\label{thm:nonHS_extension}
Assume that $K\in\mathbb K$, that the induced operator $\mathbf K$ satisfies
Assumptions~\ref{ass:dominant_spectral_separation} and
\ref{ass:reference_generating_nonresonance}.  Then the series
\[
w_\rho(x,y)
=
c_0\sum_{n=1}^{\infty}
\frac{\Gamma_n(x,y)}{\rho(\mathbf K)^n}
\]
converges pointwise absolutely and represents a nontrivial dominant
eigenfunction.  Its coefficients satisfy the Bell-polynomial formula
\eqref{eq:HS_gamma}.  Thus the coefficient construction identified in
Proposition~\ref{prop:HS_eigenfunction} extends conditionally to kernels in
\(\mathbb K\), without use of a Fredholm determinant.
\end{theorem}

\begin{proof}
Lemma~\ref{lem:Gamma_Bell_nonHS} gives the coefficient formula for every
$K\in\mathbb K$.  By
Corollary~\ref{cor:gamma_normalized_convergence}, the series
\[
\sum_{n=1}^{\infty}\frac{\Gamma_n(x,y)}{\rho(\mathbf K)^n}
\]
converges absolutely at the dominant spectral value. Proposition~\ref{prop:dominant_eigenvalue_riesz} therefore yields the representation
\[
w_\rho(x,y)
=
c_0\sum_{n=1}^{\infty}\frac{\Gamma_n(x,y)}{\rho(\mathbf K)^n}.
\]
The Bell-polynomial expression follows from
Lemma~\ref{lem:Gamma_Bell_nonHS}.
\end{proof}

\begin{remark}[Scope of the extension]
\label{rem:scope_extension}
\rm
Theorem~\ref{thm:nonHS_extension} should be read as a conditional extension
result. The kernel class $\mathbb K$ ensures that the reference-point
construction is meaningful at the kernel level, while
Assumptions~\ref{ass:dominant_spectral_separation} and
\ref{ass:reference_generating_nonresonance} supply, respectively, the spectral
separation and scalar denominator control needed for convergence at the Perron
root. In particular, the theorem does not assert that every kernel in
$\mathbb K$ is quasi-compact or admits such a decomposition.  The extension is
therefore an extension of the representation to some non-Hilbert--Schmidt
kernels, not a set-theoretic inclusion of every kernel covered by classical
Fredholm theory.  Conversely, a Hilbert--Schmidt kernel need not satisfy the
pointwise spectral bound or the reference condition without further
hypotheses.
\end{remark}
\section{Application to simple integral projection models at discrete time}
\label{sec:discrete-ipm}

In this section, we investigate the eigenvalue problem for the simplest discrete-time IPM, using the eigenfunction representation established in Section~\ref{sec:fredholm}. This representation also yields a Markovian viewpoint and facilitates biological interpretation.

\paragraph{Role of this section.}
\label{para:role_section3}

The purpose of this section is to specialize the reference-point construction of Section~\ref{sec:fredholm} to the discrete-time IPM setting and to interpret the resulting eigensystem in genealogical terms. Section~\ref{sec:fredholm} provides a constructive representation of the eigenfunction through the $\Gamma_n$-series under explicit spectral separation and reference non-resonance assumptions. In the present section, we apply that representation to the simplest discrete-time IPM, derive the asymptotic profile of the cohort, and reinterpret the resulting quantities by analogy with taboo probabilities in Markov chains.

The simplest kernel to keep in mind is
\[
K(x,y)=S(x,y)+F(x,y),
\]
where \(S(x,y)\) describes survival and transition from state \(y\) to state
\(x\), and \(F(x,y)\) describes the production of offspring in state \(x\) by
an individual in state \(y\).  The results below do not depend on this
particular decomposition, but it explains why iterated kernels \(K^{(n)}\)
have a natural genealogical reading: they collect contributions over \(n\)
successive transition or reproductive steps.
Thus the direct-contribution terms below measure how a focal initial state
feeds the stable population and reproductive value through repeated survival,
transition, and reproduction events.  This is the quantity that is typically
hidden when the continuous kernel is replaced immediately by a quadrature
matrix.

Let $P_t(x)$ denote the cohort density at state $x$ at time $t$. We consider the IPM
\begin{equation}\label{sipm1}
P_{t+1}(x)=\int_{\Omega^d}K(x,y)\,P_t(y)\,dy,
\qquad
P_0\in L^1(\Omega^d),\quad P_0\ge 0.
\end{equation}
Here the kernel $K\in\mathbb K$ is assumed to satisfy the structural assumptions introduced in Section~\ref{sec:fredholm}.

\begin{assumption}[Discrete-time spectral separation and reference non-resonance]
\label{assump:discrete_ipm}
Let $\mathbf K:L^1(\Omega^d)\to L^1(\Omega^d)$ be the integral operator induced by the kernel $K\in\mathbb K$. Assume that:
\begin{enumerate}
\item[(i)] $\lambda_0=\rho(\mathbf K)>0$ is an isolated algebraically simple eigenvalue of $\mathbf K$;
\item[(ii)] there exist a rank-one Riesz projection with kernel
\(\Pi(x,y)\),
a number $\theta\in(0,\lambda_0)$, and a function $C_1\in\mathcal X$ with $C_1(x,y)>0$ such that
\begin{equation}\label{eq:kernel_iterate_riesz_discrete}
K^{(n)}(x,y)=\lambda_0^n\Pi(x,y)+R_n(x,y),
\qquad
|R_n(x,y)|\le C_1(x,y)\theta^n,
\qquad n\ge1;
\end{equation}
\item[(iii)] we choose a diagonal reference point $y_0\in\Omega^d$ such that
\begin{equation}\label{eq:diag_ref_nonzero}
\kappa_0:=\Pi(y_0,y_0)>0,
\end{equation}
and normalize the right and left eigenfunctions at \(y_0\), equivalently
defining \(U:=\Pi/\kappa_0=u(x)v(y)\), so that
\begin{equation}\label{eq:diag_ref_normalized}
U(y_0,y_0)=1.
\end{equation}
\item[(iv)] with \(r_n:=R_n(y_0,y_0)\) and
\[
E(t):=\sum_{n=1}^{\infty}r_nt^n,
\]
there exists \(r_\ast\) such that
\[
\lambda_0^{-1}<r_\ast<\theta^{-1}
\]
and
\begin{equation}\label{eq:discrete_reference_denominator}
1+(\kappa_0-1)\lambda_0t+(1-\lambda_0 t)E(t)\neq0
\qquad (|t|\le r_\ast).
\end{equation}
\end{enumerate}
\end{assumption}

This is the diagonal-reference, model-specific restatement of the two
assumptions used in Section~\ref{sec:fredholm}: items (i)--(ii) give dominant
spectral separation, item (iii) fixes the normalized reference pair, and item
(iv) is reference non-resonance. The additional mixed-norm bound on $C_1$ is
satisfied on bounded state spaces and, more generally, when the transition and
reproduction kernels admit a uniform integrable envelope. It therefore retains
the standard finite-state, bounded-trait, and uniformly light-tailed IPM
settings used in demographic applications.

Under Assumption~\ref{assump:discrete_ipm}, Corollary~\ref{cor:gamma_normalized_convergence} applies with the reference pair $(y_0,y_0)$. Hence the reference-point series
\[
\sum_{n\ge1}\lambda^{-n}\Gamma_n(\cdot,\cdot;y_0,y_0)
\]
converges absolutely at $\lambda=\lambda_0=\rho(\mathbf K)$.

Assumption~\ref{assump:discrete_ipm} is imposed to keep the subsequent formulas both spectrally natural and biologically interpretable. In particular, the diagonal choice $y_0$ allows us to present the direct-contribution formulas in a transparent form. Without this convention, the same arguments remain valid, but the notation becomes more cumbersome because the three auxiliary points in the reference-point construction must be kept distinct.

\begin{remark}[A biological interpretation of the diagonal reference choice]
\rm
The normalization at a diagonal reference point is natural in many biological settings, for instance for long-lived organisms with slow growth such as trees, for species whose adult size is essentially fixed (as is often the case in mammals), and for organisms exhibiting strong site fidelity, that is, individuals that remain in the same habitat or patch with little movement.
\end{remark}

Before formulating the eigenvalue problem for \eqref{sipm1}, we introduce a convention regarding the reference points appearing in Section~\ref{sec:fredholm}. The representation of the eigenfunction there involves three auxiliary points, denoted by $y$, $x_0$, and $y_0$, which may be chosen freely as long as the associated reference pair satisfies the assumptions of Corollary~\ref{cor:gamma_normalized_convergence}. However, in view of the biological interpretation discussed below, we impose the convention
\[
y=x_0=y_0.
\]
Thus we fix a diagonal reference point $y_0$ satisfying \eqref{eq:diag_ref_nonzero}--\eqref{eq:diag_ref_normalized} and suppress the repeated variables from the notation as follows:
\begin{align*}
w_{0}(x,y_0)&=w_{0}(x,y_0,y_0,y_0),\\
v_{0}(y_0,x)&=v_{0}(y_0,y_0,y_0,x),\\
\Gamma_n(x,y_0)&=\Gamma_n(x,y_0;y_0,y_0),\\
\Gamma_n^{*}(y_0,x)&=\Gamma_n^{*}(y_0,y_0;y_0,x).
\end{align*}

Under this convention, Proposition~\ref{prop:dominant_eigenvalue_riesz} yields the eigenfunction $w_0(x,y_0)$ corresponding to \eqref{sipm1} in the form
\begin{align}
w_{0}(x,y_0)
&=
c_{w}(y_0)\left(\sum_{n=1}^{\infty}\frac{\Gamma_{n}(x,y_0)}{\lambda_{0}^{n}}\right),
\qquad c_{w}(y_0)\neq 0,\label{ef1}\\
\Gamma_{1}(x,y_0)
&=
K(x,y_0),\label{efi0}\\
\Gamma_{n+1}(x,y_0)
&=
\int_{\Omega^{d}}K(x,\xi)\Gamma_{n}(\xi,y_0)\,d\xi
-
K(x,y_0)\Gamma_{n}(y_0,y_0),
\qquad n\ge 1.\label{efi1}
\end{align}

For the adjoint eigenfunction $v_0(y_0,x)$, we adopt the same diagonal choice $y=x_0=y_0$, yielding
\begin{align}
v_{0}(y_0,x)
&=
c_v(y_0)\left(\sum_{n=1}^{\infty}\frac{\Gamma_{n}^{*}(y_0,x)}{\lambda_{0}^{n}}\right),
\qquad c_v(y_0)\neq 0,\label{aef2}\\
\Gamma_{1}^{*}(y_0,x)
&=
K(y_0,x),\label{aefi0}\\
\Gamma_{n+1}^{*}(y_0,x)
&=
\int_{\Omega^{d}}\Gamma_{n}^{*}(y_0,\eta)K(\eta,x)\,d\eta
-
\Gamma_{n}^{*}(y_0,y_0)K(y_0,x),
\qquad n\ge 1.\label{aefi1}
\end{align}

Note that, for notational consistency, the variables $x$ and $y$ in \eqref{aef2} have been interchanged to align with the convention of expressing functions with respect to $x$.

\subsection{Asymptotic characterization by the eigensystem}

Recall that the cohort dynamics \eqref{sipm1} can be written as
\[
P_{t}=\mathbf K^{\,t}P_{0}
\qquad (t\in\mathbb N),
\]
where $\mathbf K:L^{1}(\Omega^{d})\to L^{1}(\Omega^{d})$ is the integral operator
\[
(\mathbf K f)(x):=\int_{\Omega^{d}}K(x,y)f(y)\,dy.
\]
The dual pairing is
\begin{equation}\label{eq:pairing_x}
\langle f,g\rangle_x:=\int_{\Omega^d} f(x)g(x)\,dx,
\qquad f\in L^\infty(\Omega^d),\ g\in L^1(\Omega^d),
\end{equation}
and the adjoint operator $\mathbf K^{*}:L^{\infty}(\Omega^{d})\to L^{\infty}(\Omega^{d})$ is given by
\[
(\mathbf K^{*}v)(y)=\int_{\Omega^{d}} v(x)\,K(x,y)\,dx.
\]

\begin{theorem}[Asymptotics of the cohort]
\label{thm:asympt2}
Assume that Assumption~\ref{assump:discrete_ipm} holds. Let $w_{0}(\cdot,y_0)\in L^{1}(\Omega^{d})$ and $v_{0}(y_0,\cdot)\in L^{\infty}(\Omega^{d})$ be nontrivial eigenfunctions satisfying
\[
\mathbf K\,w_{0}(\cdot,y_0)=\lambda_{0}\,w_{0}(\cdot,y_0),
\qquad
\mathbf K^{*}v_{0}(y_0,\cdot)=\lambda_{0}\,v_{0}(y_0,\cdot),
\]
with $\langle v_{0},w_{0}\rangle_x\neq 0$. Then there exist constants $C>0$ and $\delta>0$ such that, for every $P_{0}\in L^{1}(\Omega^{d})$ with $P_{0}\ge 0$,
\begin{equation}\label{asympt2}
\left\|
P_t-
\frac{\langle v_{0},P_{0}\rangle_{x}}{\langle v_{0},w_{0}\rangle_{x}}
\lambda_{0}^{t}w_{0}(\cdot,y_0)
\right\|_{L^1}
\le C\lambda_0^t e^{-\delta t}\|P_0\|_{L^1},
\qquad t\in\mathbb N,
\end{equation}
where $\langle\cdot,\cdot\rangle_x$ is defined in \eqref{eq:pairing_x}.
\end{theorem}

\begin{proof}
Since $\lambda_0$ is assumed to be an isolated simple eigenvalue, let $\Gamma$ be a positively oriented circle in $\mathbb C$ centered at $\lambda_0$ and enclosing no other point of $\sigma(\mathbf K)$. Define the Riesz projection
\[
\mathbf L
:=
\frac{1}{2\pi i}\oint_{\Gamma}(zI-\mathbf K)^{-1}\,dz
\qquad \text{on }L^1(\Omega^d).
\]
Then $\mathbf L$ is a bounded projection commuting with $\mathbf K$ and satisfying
\[
\mathbf K\mathbf L=\lambda_0\mathbf L.
\]
Since $\lambda_0$ is simple, $\mathrm{Ran}(\mathbf L)=\mathrm{span}\{w_{0}(\cdot,y_0)\}$. Hence there exists a bounded linear functional $\ell$ on $L^{1}(\Omega^{d})$ such that
\[
\mathbf L f=\ell(f)\,w_{0}(\cdot,y_0),
\qquad
f\in L^{1}(\Omega^{d}).
\]
Because $\mathbf L$ commutes with $\mathbf K$, its adjoint $\mathbf L^{*}$ commutes with $\mathbf K^{*}$. Moreover, $\mathbf K^{*}v_{0}=\lambda_{0}v_{0}$ implies $\mathbf L^{*}v_{0}=v_{0}$. Therefore, for every $f\in L^{1}(\Omega^{d})$,
\[
\langle v_{0},\mathbf L f\rangle_x
=
\langle \mathbf L^{*}v_{0},f\rangle_x
=
\langle v_{0},f\rangle_x.
\]
On the other hand, $\mathbf L f=\ell(f)w_{0}$ gives
\[
\langle v_{0},\mathbf L f\rangle_x
=
\ell(f)\,\langle v_{0},w_{0}\rangle_x.
\]
Hence
\[
\ell(f)=\frac{\langle v_{0},f\rangle_x}{\langle v_{0},w_{0}\rangle_x},
\qquad
\mathbf L f=\frac{\langle v_{0},f\rangle_x}{\langle v_{0},w_{0}\rangle_x}\,w_{0}.
\]

Now define
\[
\mathbf R:=\mathbf K(\mathbf I-\mathbf L).
\]
Then
\[
\mathbf L\mathbf R=\mathbf R\mathbf L=0,
\qquad
\mathbf K^t=\lambda_0^t\mathbf L+\mathbf R^t
\qquad (t\in\mathbb N).
\]
By Assumption~\ref{assump:discrete_ipm}(ii), the kernel of $\mathbf R^t$ is $R_t$ and hence
\[
\|\mathbf R^{t}\|_{L^{1}\to L^{1}}
\le
\|C_1\|_{\mathcal X}\theta^t
=
\|C_1\|_{\mathcal X}\lambda_0^t e^{-\delta t},
\qquad
\delta:=\log(\lambda_0/\theta)>0.
\]
Applying this decomposition to $P_{t}=\mathbf K^{t}P_{0}$ yields
\[
\left\|P_t-\lambda_0^t\mathbf L P_0\right\|_{L^1}
\le
\|C_1\|_{\mathcal X}\lambda_0^t e^{-\delta t}\|P_0\|_{L^1},
\]
which gives \eqref{asympt2}.
\end{proof}

\subsection{Reinterpretation of eigensystems by analogy with Markov chains}

In understanding the sequence appearing in \eqref{ef1}, the theory of Markov chains offers particularly valuable suggestions. In a Markov chain whose transition probability from state $j$ to state $i$ is denoted by $p_{ij}\ge0$, the following quantity, called the taboo probability, is known \cite{chung1960transition}:
\begin{align}
p^j_{ij}(n)
&:=
\mathbb P_j\!\left(X_n=i,\ X_k\neq j\ \text{for all }1\le k\le n-1\right)\nonumber\\
&=
\sum_{i_1\neq j}\sum_{i_2\neq j}\cdots\sum_{i_{n-1}\neq j}
p_{ii_0}p_{i_0i_1}\cdots p_{i_{n-1}j}.
\end{align}
It is well known that the following sequence constitutes the stationary distribution $\mu(i)$ of this Markov process:
\begin{equation}
\mu(i)=p_{ij}+\sum_{n=2}^{\infty}p^j_{ij}(n).
\end{equation}
If $(p_{ij})_{1\le i,j\le M}$ is an irreducible stochastic matrix, the stationary distribution $(\mu(i))_{1\le i\le M}$ can be equivalently described as the eigenvector corresponding to the largest eigenvalue $1$. Considering the recurrence relation that the $n$-step taboo probabilities satisfy, one formally obtains
\begin{equation}\label{sd1}
p^j_{ij}(n)=\sum_{k=1}^{M}p_{ik}p^j_{kj}(n-1)-p_{ij}p^j_{jj}(n-1),
\qquad
p^j_{ij}(1)=p_{ij}.
\end{equation}
This relation does not require the matrix to be stochastic; an irreducible nonnegative matrix similarly yields an eigenvector corresponding to its Frobenius root \cite{oizumi2022sensitivity}. Focusing on the right-hand side of \eqref{sd1}, we see that the first term sums over all paths from every state $k$ to state $i$ at the previous step, while the second term subtracts the contribution of the paths that pass through state $j$. Replacing the sum with an integral in \eqref{sd1}, we observe that the resulting relation resembles the recursion \eqref{efi1}.

However, from a measure-theoretic viewpoint, since a single point has Lebesgue measure zero, the expression \eqref{efi1} cannot literally be interpreted as ``subtracting the paths passing through $y_0$ from all paths leading to $x$.'' Therefore, when $w(y_0)=1$, we define the series on the right-hand side of \eqref{ef1} as the \emph{direct contribution} from $y_0$ to $x$. Similarly, we define the right-hand side of \eqref{aef2} as the \emph{adjoint direct contribution} from $y_0$ to $x$. These two direct contributions respectively represent the degree of contribution from a past state $y_0$ to a future state $x$, and the degree of dependence of a future state $y_0$ on a past state $x$. Furthermore, we define the \emph{self-direct contribution} as the direct contribution from a state $y_0$ to itself, where the direct contribution and its adjoint coincide. In a Markov process, a self-direct contribution of one indicates recurrence; in the IPM, the value of $\lambda_0$ that makes the self-direct contribution equal to one gives the intrinsic growth rate:
\begin{equation}\label{ele1}
\sum_{n=1}^{\infty}\frac{\Gamma_{n}(y_0,y_0)}{\lambda_{0}^{n}}
=
\sum_{n=1}^{\infty}\frac{\Gamma_{n}^{*}(y_0,y_0)}{\lambda_{0}^{n}}
=
1.
\end{equation}
In finite-dimensional models, namely transition matrix models, the self-direct contribution indeed reflects its name: it sums, over each number of steps, the paths that return to the same state for the first time.

\subsection{Initial population dependence and expected contribution steps}

Building on the previous subsection, the numerator $\langle v_0, P_0 \rangle_{x}$ of the expansion coefficient in \eqref{asympt2}---the pairing of the reproductive value $v_0$ with the initial population $P_0$---quantifies the dependence of a future state $y_0$ on the initial distribution $P_0(\cdot)$:
\begin{equation}
\langle v_0, P_0 \rangle_{x}
=
c_v(y_0)\sum_{n=1}^{\infty}\int_{\Omega^{d}}\frac{\Gamma_{n}^{*}(y_0,x)\,P_0(x)}{\lambda_{0}^{n}}\, dx.
\end{equation}
\begin{theorem}\label{th2}
Let $\Gamma_{m}^{*}(y_0,x)$ satisfy \eqref{aefi1} and $\Gamma_{n}(x,y_0)$ satisfy \eqref{efi1}. Then
\begin{equation}\label{rel1}
\langle \Gamma_{m}^{*}, \Gamma_{n} \rangle_{x}
=
\Gamma_{m+n}(y_0,y_0)+\Gamma_{m}(y_0,y_0)\Gamma_{n}(y_0,y_0),
\qquad m,n\ge 1.
\end{equation}
\end{theorem}

\begin{proof}
Set
\[
\Gamma_{n}:=\Gamma_{n}(y_0,y_0)=\Gamma^{*}_{n}(y_0,y_0),
\qquad n\ge 1.
\]
For integers $m\ge 1$ and $n\ge 1$, define
\begin{equation}\label{eq:def-phi}
\phi(m,n)
:=
\langle \Gamma_{m}^{*}, \Gamma_{n} \rangle_{x}
-
\Gamma_{m+n}
-
\Gamma_{m}\Gamma_{n}.
\end{equation}
Equations \eqref{efi1} and \eqref{aefi1} give
\begin{align}
\langle \Gamma_{m}^{*}, \Gamma_{n} \rangle_{x}
&=
\int_{\Omega^d}
\left(
\int_{\Omega^d}\Gamma_{m-1}^{*}(y_0,\eta)K(\eta,x)\,d\eta
-
\Gamma_{m-1}K(y_0,x)
\right)
\Gamma_{n}(x,y_0)\,dx \nonumber\\
&=
\langle \Gamma_{m-1}^{*}, \Gamma_{n+1} \rangle_{x}
+
\int_{\Omega^d}\Gamma_{m-1}^{*}(y_0,\eta)K(\eta,y_0)\Gamma_{n}\,d\eta \nonumber\\
&\quad
-
\Gamma_{m-1}\int_{\Omega^d}K(y_0,x)\Gamma_{n}(x,y_0)\,dx \nonumber\\
&=
\langle \Gamma_{m-1}^{*}, \Gamma_{n+1} \rangle_{x}
+
\Gamma_{m}\Gamma_{n}
-
\Gamma_{m-1}\Gamma_{n+1}.\label{in1}
\end{align}
Substituting \eqref{in1} into \eqref{eq:def-phi} yields
\[
\phi(m,n)=\phi(m-1,n+1),
\qquad m,n\ge1.
\]
Iterating this identity $m-1$ times gives $\phi(m,n)=\phi(1,m+n-1)$.
For \(m=1\), since
\(\Gamma_1^*(y_0,x)=K(y_0,x)\), the recursion \eqref{efi1} evaluated at
\((y_0,y_0)\) gives
\[
\Gamma_{k+1}
=
\int_{\Omega^d}K(y_0,x)\Gamma_k(x,y_0)\,dx
-
K(y_0,y_0)\Gamma_k .
\]
Using \(K(y_0,y_0)=\Gamma_1\), we therefore obtain
\[
\langle \Gamma_1^*,\Gamma_k\rangle_x
=
\int_{\Omega^d}K(y_0,x)\Gamma_k(x,y_0)\,dx
=
\Gamma_{k+1}+\Gamma_1\Gamma_k .
\]
Consequently
\[
\phi(1,k)
=
\langle \Gamma_1^*,\Gamma_k\rangle_x-\Gamma_{k+1}-\Gamma_1\Gamma_k
=0,
\qquad k\ge1.
\]
Hence $\phi(m,n)=0$ for all $m,n\ge 1$.
\end{proof}

\begin{remark}
\rm
The condition $y=x_{0}=y_{0}$ imposed in Theorem~\ref{th2} is essential. If, instead, one keeps the general reference-point setting of Section~\ref{sec:fredholm} and allows the variables in the direct contribution and its adjoint to vary independently, additional summation terms appear on the right-hand side of \eqref{rel1}. Such terms not only obscure the biological interpretation of the eigensystem but also considerably complicate the associated computations. It is also worth noting that imposing $y=x_{0}=y_{0}$ alters the result at most by a multiplicative constant.
\end{remark}

In the following calculation we assume that the two eigenfunction series may
be paired termwise; for example, it is sufficient that they converge
absolutely in $L^1$ and $L^\infty$, respectively.  By invoking
Theorem~\ref{th2}, $\langle v_{0},w_{0}\rangle$ is then computed as follows:
\begin{align}
\langle v_{0},w_{0}\rangle_{x}
&=
\int_{\Omega^d}v_{0}(y_0,x)\,w_{0}(x,y_0)\,dx \nonumber\\
&=
c_v(y_0)c_w(y_0)\int_{\Omega^d}
\left(\sum_{m=1}^{\infty}\frac{\Gamma_{m}^{*}(y_0,x)}{\lambda_{0}^{m}}\right)
\left(\sum_{n=1}^{\infty}\frac{\Gamma_{n}(x,y_0)}{\lambda_{0}^{n}}\right)\,dx \nonumber\\
&=
c_v(y_0)c_w(y_0)\sum_{m=1}^{\infty}\sum_{n=1}^{\infty}\frac{1}{\lambda_{0}^{m+n}}
\int_{\Omega^d}\Gamma_{m}^{*}(y_0,x)\Gamma_{n}(x,y_0)\,dx \nonumber\\
&=
c_v(y_0)c_w(y_0)\sum_{m=1}^{\infty}\sum_{n=1}^{\infty}\frac{1}{\lambda_{0}^{m+n}}
\left(\Gamma_{m+n}(y_0,y_0)+\Gamma_{m}(y_0,y_0)\Gamma_{n}(y_0,y_0)\right) \nonumber\\
&\quad \text{(by Theorem~\ref{th2})} \nonumber\\
&=
c_v(y_0)c_w(y_0)\left(
\underbrace{\sum_{k=2}^{\infty}\frac{(k-1)\Gamma_k(y_0,y_0)}{\lambda_0^{k}}}_{(\ast)}
+
\underbrace{\left(\sum_{n=1}^{\infty}\frac{\Gamma_{n}(y_0,y_0)}{\lambda_{0}^{n}}\right)^2}_{=\,1\text{ by }\eqref{ele1}}
\right).\label{es1}
\end{align}
Define the expected number of contributing steps by
\begin{equation}\label{em}
\mathbb{E}_{y_0}[m]
:=
\sum_{n=1}^{\infty}n\,\frac{\Gamma_{n}(y_0,y_0)}{\lambda_{0}^{n}}.
\end{equation}
Then
\[
(\ast)=\mathbb E_{y_0}[m]-\sum_{n\ge1}\frac{\Gamma_n(y_0,y_0)}{\lambda_0^n}
=
\mathbb E_{y_0}[m]-1,
\]
hence \eqref{es1} becomes
\begin{equation}\label{es1b}
\langle v_{0},w_{0}\rangle_{x}
=
c_v(y_0)c_w(y_0)\,\mathbb{E}_{y_0}[m].
\end{equation}

\begin{remark}
\rm
The individual coefficients in \eqref{em} need not be nonnegative.  Nevertheless,
choosing the positive eigenfunctions with $c_v(y_0),c_w(y_0)>0$ and using
\eqref{es1b} gives
\[
\mathbb E_{y_0}[m]
=\frac{\langle v_0,w_0\rangle_x}{c_v(y_0)c_w(y_0)}>0.
\]
We therefore call this quantity the expected number of contributing steps by
analogy with the mean recurrence time, not because its individual summands
form a probability distribution.  Its precise demographic interpretation is
model dependent.

Substituting \eqref{ef1}, \eqref{rel1}, and \eqref{es1b} into \eqref{asympt2} yields, in particular under the leading eigenvalue condition $\lambda_0=1$, demographic coefficients that determine the steady-state total population size:
\begin{equation}
\lim_{t\to +\infty}\int_{\Omega^d} P_{t}(x)\,dx
=
\frac{1}{\mathbb{E}_{y_0}[m]}
\sum_{n=1}^{\infty}\sum_{k=1}^{\infty}
\int_{\Omega^{d}}\Gamma_{k}(x,y_0)\,dx
\int_{\Omega^{d}}\Gamma_{n}^{*}(y_0,\xi)P_0(\xi)\,d\xi.
\end{equation}
These coefficients represent the reproductive contribution of the initial population at state $y_0$, multiplied by the total direct contribution from state $y_0$ to all other states $x$, and normalized by the expected number of transition steps. Equivalently, the total population can be interpreted as the product of the expected reproduction and survival for the cohort at the initial state $y_0$ and the per-step contribution rate at state $y_0$. A larger expected step count $\mathbb{E}_{y_0}[m]$ implies a smaller per-step contribution of descendants. Since $\lambda_0=1$ corresponds to population replacement, \eqref{ele1} yields
\[
\sum_{n=1}^{\infty}\Gamma_{n}(y_0,y_0)=1.
\]
\end{remark}

We define
\begin{equation}
T_{y_0}:=\sum_{n=1}^{\infty}\Gamma_n(y_0,y_0).
\end{equation}
Then the following proposition holds.

\begin{proposition}
If $0<T_{y_0}\le 1$, then $0<\lambda_0\le 1$.
\end{proposition}

\begin{proof}
This follows from the fact that the right-hand side of
\[
F(\lambda)=\sum_{n=1}^{\infty}\lambda^{-n}\Gamma_n(y_0,y_0)
\]
is strictly decreasing in $\lambda$ on $(\lambda_0,\infty)$, together with the normalization identity \eqref{ele1}.
\end{proof}

\begin{remark}
\rm
The quantity $T_{y_0}$ is referred to as the \emph{type reproduction number} (TRN) \cite{heesterbeek2007type,inaba2013definition} at state $y_0$. By analogy with Markov chains, it aggregates, over all $n$, the total contribution from individuals originating in state $y_0$ who either return to $y_0$ or produce descendants that reach $y_0$ for the first time at step $n$. From a measure-theoretic viewpoint, a single point in $\mathbb{R}^d$ has zero recurrence measure; thus this interpretation is heuristic. Nevertheless, given the meaning of the quantity and its relation to the dominant spectral value $\lambda_0$, it is natural to regard $T_{y_0}$ as the type reproduction number associated with a single state.

Thus far, analytical insight has been obtained through the spectral analysis of the discrete-time IPM \eqref{sipm1}, including the characteristic equation \eqref{ele1}, the eigensystems \eqref{ef1} and \eqref{aef2}, and the construction of the type reproduction number via the reference-point representation of Section~\ref{sec:fredholm}. However, empirical IPMs often abstract away age-structured life history due to observational constraints, limiting biological interpretation. To address life-history, demographic, and evolutionary questions, it is therefore necessary to incorporate age structure explicitly into the mathematical formulation.
\end{remark}
\section{Multi-state McKendrick equation}
\label{sec:mckendrick}

\paragraph{Connection with the general framework.}
\label{para:bridge_section2_section4}

The analysis developed in Section~\ref{sec:fredholm} applies directly to the multistate McKendrick model through the Laplace-transformed kernel
\[
\psi(x,y;r)
=
\int_0^\alpha e^{-ra}
\int_{\Omega^d}
F(x\leftarrow \xi;a)\,K(a,\xi\leftarrow 0,y)\,d\xi\,da.
\]
For each fixed $r$ such that $\psi(\cdot,\cdot;r)\in\mathbb K$, this kernel defines a positive integral operator on the space $\mathcal X$. Hence the reference-point construction introduced in Section~\ref{sec:fredholm}---in particular the $\Gamma_n$ recursion and the Bell-polynomial representation---applies without modification at the kernel level.

The spectral conclusions of Section~\ref{sec:fredholm}, however, require additional spectral separation and reference non-resonance assumptions. In the present section, we first derive the age-zero renewal kernel $\psi(\cdot,\cdot;r)$ from the multistate McKendrick equation, and then impose the corresponding assumptions only at the biologically relevant parameter values. Under these assumptions, the determinant-free Fredholm framework developed in Section~\ref{sec:fredholm} yields explicit genealogical representations of the stable birth-state distribution and the reproductive value.

From a biological viewpoint, the kernel $\psi(x,y;r)$ represents the expected contribution from state $y$ at birth to state $x$ at the next generation, discounted by the exponential factor $e^{-ra}$. Therefore, the spectral problem for $\psi$ corresponds to an Euler--Lotka-type equation \cite{euler1760mortality,sharpe1911problem}, and the genealogical expansion derived in Section~\ref{sec:fredholm} provides a multigenerational decomposition of reproductive contributions.

This connection justifies applying the abstract theory of Section~\ref{sec:fredholm} to the concrete demographic model introduced below.

\begin{remark}[Interpretation of $\psi$ as a next-generation kernel]
\label{rem:psi_next_generation}
\rm
The kernel $\psi(x,y;r)$ admits a natural interpretation as a next-generation operator. For a given growth rate $r$, the quantity
\[
\psi(x,y;r)
=
\int_0^\alpha e^{-ra}
\int_{\Omega^d}
F(x\leftarrow \xi;a)\,K(a,\xi\leftarrow 0,y)\,d\xi\,da
\]
represents the expected contribution to individuals in state $x$ at birth from a single individual initially in state $y$, aggregated over all ages and discounted by $e^{-ra}$.

Thus, $\psi(\cdot,\cdot;r)$ plays the role of a reproduction kernel across generations. In particular, the spectral radius $\rho(\hat{\Psi}(r))$ corresponds to a type reproduction number in a weighted state-structured sense. The characteristic equation
\[
\rho(\hat{\Psi}(r))=1
\]
therefore plays the role of an Euler--Lotka condition determining the intrinsic growth rate $r$.

Within the reference-point framework, this condition is refined into the genealogical identity
\[
\sum_{n=1}^\infty \Gamma_n(y_0,y_0;r)=1
\]
at the normalized reference point, where the terms $\Gamma_n$ represent multigenerational contributions. In this sense, the present construction may be viewed as a genealogical refinement of the classical next-generation operator approach.
\end{remark}

\subsection{Assumptions of the multi-state McKendrick equation}

\noindent\textbf{Multi-state age-structured IPM.}
We consider the transition kernel
\[
K : [0,\alpha) \times \Omega^d \times [0,\alpha) \times \Omega^d \to [0,\infty),\quad
(a,x,s,y) \mapsto K(a, x \leftarrow s, y),
\]
representing the probability density of transitioning from state $y$ at age $s$ to state $x$ at age $a$. We assume:
\begin{itemize}
\item $K(a,x\leftarrow s,y)\in\mathbb K$ for $a>s$;
\item for all $N\in\mathbb N$ and $a>s$,
\begin{equation}
\sup_{y\in\Omega^d}|x|^N K(a,x\leftarrow s,y)\to0
\quad\text{as }|x|\to\infty;
\end{equation}
\item $\lim_{a\downarrow s}K(a,x\leftarrow s,y)=\delta^d(x-y)$;
\item $K(a,x\leftarrow s,y)=0$ for $a<s$;
\item (Chapman--Kolmogorov equation) for $s\le\tau\le a<\alpha$,
\begin{equation}\label{eq:mck_ck}
K(a,x\leftarrow s,y)
=
\int_{\Omega^d}K(a,x\leftarrow \tau,z)\,K(\tau,z\leftarrow s,y)\,dz;
\end{equation}
\item (Boundary) $K(\alpha,x\leftarrow s,y)=0$.
\end{itemize}
Note that $K(a,x\leftarrow s,y)$ is a (sub-)Markov transition density in the sense that
\begin{equation}
\int_{\Omega^d}K(a,x\leftarrow s,y)\,dx\le1.
\end{equation}

Let $P_t(a,x)$ denote the age-state density of a population at time $t\in[0,\infty)$, where $a\in[0,\alpha)$ is chronological age and $x\in\Omega^d\subseteq\mathbb R^d$ is a $d$-dimensional state variable. We define the multi-state age-structured IPM governed by $K$ by
\begin{equation}\label{eq:mck_ipm}
P_{t+\varepsilon}(a+\varepsilon,x)
=
\int_{\Omega^d}K(a+\varepsilon,x\leftarrow a,y)\,P_t(a,y)\,dy,
\qquad \varepsilon>0,
\end{equation}
with initial condition
\begin{equation}\label{eq:mck_initial}
P_0(a,x)=\varphi(a,x)\in L^1\big([0,\alpha)\times\Omega^d\big),
\qquad
\varphi(a,x)>0.
\end{equation}

To construct a renewal equation from \eqref{eq:mck_ipm}, we introduce a fertility function
\[
F:[0,\alpha)\times\Omega^d\times\Omega^d\to[0,\infty)
\]
that yields the inhomogeneous birth rate in age and state. We assume:
\begin{itemize}
\item $F(x\leftarrow y;a)>0$;
\item for each $a\in[0,\alpha)$, $F(x\leftarrow y;a)$ is measurable and continuous in $(x,y)\in\Omega^d\times\Omega^d$;
\item for each fixed $(a,y)$, $F(\cdot\leftarrow y;a)\in L^1(\Omega^d)$;
\item if $\Omega^{d}$ is (partially) unbounded, let $\Omega^\ell\subseteq\Omega^d$ be an unbounded subset ($0\le\ell\le d$). There exist constants $C>0$, $\beta\in\mathbb R$, and $m_i\ge0$ such that
\begin{equation}\label{eq:mck_F_growth}
F(x\leftarrow y;a)
\le
C\,e^{\beta a}\left(1+\sum_{i=1}^{\ell}|y_i|^{m_i}\right),
\end{equation}
where $y_i$ denotes the $i$th entry of $y\in\Omega^\ell$.
\end{itemize}
The generation of newborns is formulated by
\begin{equation}\label{eq:mck_birth}
P_t(0,x)
=
\int_0^\alpha\int_{\Omega^d}F(x\leftarrow y;a)\,P_t(a,y)\,dy\,da.
\end{equation}

\begin{remark}[On structural assumptions]
\label{rem:mck_assumption_structure}
\rm
In earlier formulations of multistate McKendrick-type models, monotonicity assumptions on the transition kernel were sometimes imposed for technical convenience. Such assumptions are not required for the present analysis.

In this paper, we do not assume monotonicity of the kernel $K(a,x\leftarrow s,y)$. All arguments rely instead on positivity, integrability conditions ensuring that $\psi(\cdot,\cdot;r)\in\mathbb K$, and explicit spectral separation and reference non-resonance assumptions imposed later on the age-zero renewal kernel. Thus monotonicity is not part of the standing assumptions of the theory.
\end{remark}

\subsection{Renewal equation and the role of the Laplace transform}

By the Chapman--Kolmogorov equation \eqref{eq:mck_ck}, \eqref{eq:mck_ipm} rewrites as
\begin{equation}\label{eq:mck_pt}
P_t(a,x)=
\begin{cases}
\displaystyle
\int_{\Omega^d}K(a,x\leftarrow a-t,\xi)\,\varphi(a-t,\xi)\,d\xi,
& \text{if } a\ge t,\\[2mm]
\displaystyle
\int_{\Omega^d}K(a,x\leftarrow 0,\xi)\,P_{t-a}(0,\xi)\,d\xi,
& \text{if } a<t.
\end{cases}
\end{equation}
Substituting \eqref{eq:mck_pt} into \eqref{eq:mck_birth} yields the renewal equation \cite{feller1941integral}:
\begin{align}
P_t(0,x)
&=
G_t(x)+\int_0^t\Psi(a)P_{t-a}(0,x)\,da,\label{eq:mck_renewal}\\
G_t(x)
&:=
\int_t^\alpha\int_{\Omega^d}\int_{\Omega^d}
F(x\leftarrow \xi;a)\,K(a,\xi\leftarrow a-t,y)\,\varphi(a-t,y)\,dy\,d\xi\,da,\nonumber\\
\Psi(a)f(x)
&:=
\int_{\Omega^d}\int_{\Omega^d}
F(x\leftarrow y;a)\,K(a,y\leftarrow 0,\xi)\,f(\xi)\,d\xi\,dy,\nonumber
\end{align}
for $f\in L^1(\Omega^d)$.

The Laplace transform is introduced here not as the main asymptotic tool of the section, but as the natural device for identifying the age-zero renewal kernel that governs the dominant root. For $g\in L^1(0,\infty)$ and $r\in\mathbb C$, set
\begin{equation}
\hat g(r):=\int_0^\infty e^{-rt}g(t)\,dt.
\end{equation}
Taking Laplace transforms in \eqref{eq:mck_renewal} gives
\begin{equation}\label{eq:mck_laplace}
\hat P(0,x;r)=\hat G(x;r)+\hat\Psi(r)\hat P(0,x;r),
\end{equation}
where $\hat\Psi(r):L^1(\Omega^d)\to L^1(\Omega^d)$ is
\begin{equation}\label{eq:mck_Psihat}
\hat\Psi(r)f(x)
:=
\int_0^\alpha e^{-ra}\int_{\Omega^d}\int_{\Omega^d}
F(x\leftarrow \xi;a)\,K(a,\xi\leftarrow 0,y)\,f(y)\,dy\,d\xi\,da.
\end{equation}
Formally,
\begin{equation}\label{eq:mck_formal_resolvent}
\hat P(0,x;r)=\bigl(\mathbf I-\hat\Psi(r)\bigr)^{-1}\hat G(x;r).
\end{equation}

Define the age-zero kernel
\begin{equation}\label{eq:mck_psi_kernel}
\psi(x,y;r)
:=
\int_0^\alpha e^{-ra}\int_{\Omega^d}
F(x\leftarrow \xi;a)\,K(a,\xi\leftarrow 0,y)\,d\xi\,da.
\end{equation}
Then $\hat\Psi(r)$ is precisely the integral operator induced by $\psi(\cdot,\cdot;r)$.

\begin{remark}[Scope of the Laplace-transform argument]
\rm
The formal representation \eqref{eq:mck_formal_resolvent} is useful for identifying the age-zero renewal operator, but we do not take a full inverse Laplace expansion as the main tool of this section. Indeed, the classical pole or residue expansion of the renewal solution may diverge in general \cite{schumitzky1975operator,verduyn1989exponential,lunel1990series}. The purpose of the Laplace transform here is therefore to produce the kernel $\psi(\cdot,\cdot;r)$, to which the determinant-free reference-point theory of Section~\ref{sec:fredholm} can be applied directly.
\end{remark}

\subsection{The dominant root and its eigenstructure}

For $a\in(0,\alpha)$, set
\[
B_a(x,y):=\int_{\Omega^d}F(x\leftarrow\xi;a)K(a,\xi\leftarrow0,y)\,d\xi,
\]
so that $\psi(x,y;r)=\int_0^\alpha e^{-ra}B_a(x,y)\,da$.

\begin{proposition}\label{prop:mck_psi_in_K}
Fix $r\in\mathbb R$. Assume that $(a,x,y)\mapsto B_a(x,y)$ is jointly measurable, that $B_a$ is continuous on $\overline U$ for almost every $a$, and that there is a measurable set $A\subset(0,\alpha)$ of positive measure such that $B_a(x,y)>0$ for every $a\in A$ and all $x,y\in\Omega^d$. If
\[
\int_0^\alpha e^{-ra}\left(
\|B_a\|_{\mathcal X}+\|B_a\|_{L^\infty(\Omega^d\times\Omega^d)}
+\sup_{(x,y)\in\overline U}|B_a(x,y)|
\right)da<\infty,
\]
then the kernel $\psi(\cdot,\cdot;r)$ defined by \eqref{eq:mck_psi_kernel} belongs to $\mathbb K$.
\end{proposition}

\begin{proof}
Positivity follows from the assumed positivity of $B_a$ on a set of ages of positive measure. Moreover,
\[
\|\psi(\cdot,\cdot;r)\|_{\mathcal X}
\le\int_0^\alpha e^{-ra}\|B_a\|_{\mathcal X}\,da,
\qquad
\|\psi(\cdot,\cdot;r)\|_{L^\infty}
\le\int_0^\alpha e^{-ra}\|B_a\|_{L^\infty}\,da.
\]
Both bounds are finite. The last term in the assumed integral is a dominating function on $\overline U$, so dominated convergence gives continuity there. These are precisely the defining properties of $\mathbb K$.
\end{proof}

\begin{remark}[Sufficient conditions for $\psi\in\mathbb K$]
\label{rem:mck_psi_conditions}
\rm
The preceding condition can be verified directly from age-dependent mixed-norm, uniform, and local bounds on the composite reproduction kernel $B_a$; it need not be deduced from a single global pointwise bound on $F$ and $K$. We impose it only at the parameter values used below: at $r_0$ for the stable eigensystem, at $0$ for cohort quantities, and on a neighbourhood of either value when derivatives are required. Thus stronger demographic conclusions may be obtained by adding the corresponding moment and spectral assumptions without restricting the basic construction to $r>\beta$.
\end{remark}

The condition $\psi(\cdot,\cdot;r)\in\mathbb K$ alone is not enough to apply the convergence theory of Section~\ref{sec:fredholm}. We therefore impose the corresponding spectral separation and reference non-resonance assumptions at the relevant root.

\begin{assumption}[Spectral separation and reference non-resonance for the renewal kernel]
\label{ass:mck_dominant_separation}
There exists $r_0>\beta$ at which the hypotheses of Proposition~\ref{prop:mck_psi_in_K} hold and such that the operator $\hat\Psi(r_0)$ induced by $\psi(\cdot,\cdot;r_0)$ satisfies:
\begin{enumerate}
\item[(i)] $1=\rho(\hat\Psi(r_0))$ is an isolated algebraically simple eigenvalue;
\item[(ii)] there exist a rank-one Riesz projection with kernel
\(\Pi_{r_0}(x,y)\),
a number $\theta_\psi\in(0,1)$, and a function $C_\psi\in\mathcal X$ with $C_\psi(x,y)>0$ such that
\begin{equation}\label{eq:mck_psi_iterates}
\psi^{(n)}(x,y;r_0)
=
\Pi_{r_0}(x,y)+R_n^\psi(x,y),
\qquad
|R_n^\psi(x,y)|\le C_\psi(x,y)\theta_\psi^n,
\qquad n\ge1;
\end{equation}
\item[(iii)] the reference point $y_0$ is chosen so that
\(\kappa_\psi:=\Pi_{r_0}(y_0,y_0)>0\), and we define the normalized
eigenfunction product \(U_{r_0}:=\Pi_{r_0}/\kappa_\psi\), for which
\begin{equation}\label{eq:mck_ref_normalization}
U_{r_0}(y_0,y_0)=1.
\end{equation}
\item[(iv)] with \(r_n^\psi:=R_n^\psi(y_0,y_0)\) and
\[
E_\psi(t):=\sum_{n=1}^{\infty}r_n^\psi t^n,
\]
there exists \(r_\ast\) such that
\[
1<r_\ast<\theta_\psi^{-1}
\]
and
\begin{equation}\label{eq:mck_reference_denominator}
1+(\kappa_\psi-1)t+(1-t)E_\psi(t)\neq0
\qquad (|t|\le r_\ast).
\end{equation}
\end{enumerate}
\end{assumption}

\begin{remark}
\rm
Assumption~\ref{ass:mck_dominant_separation} restates, for the age-zero renewal
operator, the two assumptions used in Section~\ref{sec:fredholm}: items (i)--(ii)
provide dominant spectral separation, item (iii) fixes the normalized reference
pair, and item (iv) is reference non-resonance. It separates the analytic issue
$\psi(\cdot,\cdot;r)\in\mathbb K$ from the spectral and scalar denominator
conditions needed for convergence at the dominant root. The corresponding
mixed-norm and exponential-moment conditions are natural in demographic
applications. In particular, the exponential-moment requirement follows from
the preceding bound when the maximal age is finite, and it also covers
infinite-age models with sufficiently light survival and reproduction tails.
\end{remark}

Under Assumption~\ref{ass:mck_dominant_separation}, define $\psi^{(1)}:=\psi$ and, for $n\ge1$,
\[
\psi^{(n+1)}(x,y;r_0)
:=
\int_{\Omega^d}\psi(x,\xi;r_0)\psi^{(n)}(\xi,y;r_0)\,d\xi.
\]
Define, for $n\ge1$,
\begin{align}
\psi_n(x,y;r_0)
:=\ &\psi^{(n)}(x,y;r_0)\nonumber\\
&+
\sum_{\ell=1}^{n-1}(-1)^\ell\sum_{k=\ell}^{n-1}
\psi^{(n-k)}(x,y;r_0)\,
\widehat B_{k,\ell}\!\bigl(\psi^{(1)},\psi^{(2)},\dots,\psi^{(k)};y,y\bigr)\Bigr|_{r=r_0},
\label{eq:mck_psin}\\
\psi_1(x,y;r_0)&:=\psi(x,y;r_0).\nonumber
\end{align}

\begin{proposition}\label{prop:mck_dominant_root}
Suppose there exists $r_0\in\mathbb R$ such that $\rho(\hat\Psi(r_0))=1$ and Assumption~\ref{ass:mck_dominant_separation} holds. Assume also that, for some $\varepsilon>0$, the hypotheses of Proposition~\ref{prop:mck_psi_in_K} hold at $r_0-\varepsilon$. Then $r_0$ is a simple root of the real equation
\[
\rho(\hat\Psi(r))=1,
\]
and $1\notin\sigma(\hat\Psi(z))$ for every $z\in\mathbb C$ with $\Re z>r_0$. Thus $r_0$ is a rightmost characteristic value.
\end{proposition}

\begin{proof}
The additional integrability assumption makes $z\mapsto\hat\Psi(z)$ holomorphic near $r_0$. Let $w_0$ and $v_0$ be positive right and left eigenfunctions at $r_0$. Since $1$ is an isolated simple eigenvalue, it has a local analytic eigenvalue branch $\lambda(r)$ with $\lambda(r_0)=1$. If $\hat B_a$ denotes the operator induced by $B_a$, then
\[
\lambda'(r_0)
=
\frac{\langle v_0,\hat\Psi'(r_0)w_0\rangle_x}
{\langle v_0,w_0\rangle_x}
=-
\frac{\displaystyle\int_0^\alpha a e^{-r_0a}
\langle v_0,\hat B_aw_0\rangle_x\,da}
{\langle v_0,w_0\rangle_x}<0.
\]
Thus $r_0$ is a simple root. Moreover, $C_\psi\in\mathcal X$ gives an operator-norm spectral gap at $r_0$, so $\rho(\hat\Psi(s))<1$ for real $s>r_0$ sufficiently close to $r_0$. Since $0\le\hat\Psi(s)\le\hat\Psi(s_1)$ whenever $s\ge s_1>r_0$, monotonicity of the spectral radius extends this inequality to every real $s>r_0$.

For $z\in\mathbb C$ with $s:=\Re z>r_0$, kernel domination gives
\[
|\hat\Psi(z)^n f|\le\hat\Psi(s)^n|f|,
\qquad n\ge1.
\]
Hence $\rho(\hat\Psi(z))\le\rho(\hat\Psi(s))<1$, and the Neumann series for $(\mathbf I-\hat\Psi(z))^{-1}$ converges. Therefore $1\notin\sigma(\hat\Psi(z))$.
\end{proof}
\begin{remark}[Reference-point characteristic equation and Euler--Lotka generalization]
\label{rem:mck_euler_lotka}
\rm
The spectral condition
\[
\rho(\hat\Psi(r_0))=1
\]
admits an explicit representation in terms of the reference-point iterates.

Indeed, evaluating the normalized expansion at the reference point $y_0$ yields
\begin{equation}\label{eq:mck_characteristic_refpoint}
\sum_{n=1}^{\infty} \psi_n(y_0,y_0;r_0) = 1.
\end{equation}
This identity can be interpreted as a characteristic equation expressed entirely in terms of genealogical contributions at the reference state.

Equation \eqref{eq:mck_characteristic_refpoint} provides a continuous-state generalization of the classical Euler--Lotka equation \cite{euler1760mortality,sharpe1911problem}. In the one-dimensional age-structured setting, the Euler--Lotka equation determines the growth rate $r$ through a balance of discounted reproduction. In the present framework, this balance is encoded at the kernel level: the total multigenerational contribution returning to the reference state equals one.

Thus, the condition $\rho(\hat\Psi(r_0))=1$ is not merely a spectral statement, but admits a concrete representation as a genealogical renewal identity.
\end{remark}

By Proposition~\ref{prop:mck_psi_in_K}, Assumption~\ref{ass:mck_dominant_separation}, and the results of Section~\ref{sec:fredholm} applied to $\psi(\cdot,\cdot;r_0)$, there exist a nonzero, nonnegative function $w_0(x,y)\in L^1(\Omega^d\times\Omega^d)$ and a nonzero, nonnegative function $v_0(y,x)\in L^\infty(\Omega^d\times\Omega^d)$ such that
\[
w_0(x,y)=\int_{\Omega^d}\psi(x,\xi;r_0)\,w_0(\xi,y)\,d\xi,
\qquad
v_0(y,x)=\int_{\Omega^d}v_0(y,\xi)\,\psi(\xi,x;r_0)\,d\xi.
\]

Choose a diagonal reference point $y_0\in\Omega^d$ such that $w_0(y_0)\neq0$. Then Proposition~\ref{prop:dominant_eigenvalue_riesz} and Corollary~\ref{cor:gamma_normalized_convergence}, applied to $\psi(\cdot,\cdot;r_0)$, show that
\begin{equation}\label{eq:mck_stable_birth}
w_0(x,y_0)=w_0(y_0)\sum_{n=1}^{\infty}\psi_n(x,y_0;r_0),
\qquad
w_0(y_0)\neq0.
\end{equation}

Similarly, define, for $m\ge1$,
\begin{align}
\psi_m^*(y_0,x;r_0)
:=\ &\psi^{(m)}(y_0,x;r_0)\nonumber\\
&+
\sum_{\ell=1}^{m-1}(-1)^\ell\sum_{k=\ell}^{m-1}
\widehat B_{k,\ell}\!\bigl(\psi^{(1)},\psi^{(2)},\dots,\psi^{(k)};y_0,y_0\bigr)\Bigr|_{r=r_0}\,
\psi^{(m-k)}(y_0,x;r_0),
\label{eq:mck_psistar}
\end{align}
and thus
\begin{equation}\label{eq:mck_adj_birth}
v_0(y_0,x)=v_0(y_0)\left(\sum_{m=1}^\infty \psi_m^*(y_0,x;r_0)\right),
\qquad
v_0(y_0)\neq0.
\end{equation}

\subsection{Asymptotics and demographic interpretation}

The main object of this section is the age-zero eigenstructure encoded by $\psi(\cdot,\cdot;r_0)$. Once the dominant root $r_0$ is determined, the corresponding stable age-state profile is propagated by the survival kernel:
\[
w_0(a,x,y_0)
:=
e^{-r_0a}\int_{\Omega^d}K(a,x\leftarrow 0,\xi)\,w_0(\xi,y_0)\,d\xi.
\]
Accordingly, \eqref{eq:mck_stable_birth} expresses the stable state distribution at age zero:
\[
w_0(0,x,y_0)=w_0(x,y_0).
\]

In the simple IPM, the direct contribution is a discrete sum over time steps. In the multistate McKendrick equation with continuous age, the kernel $\psi(x,y;r_0)$ represents the total lifetime reproductive contribution of an individual starting from state $y$ to offspring with initial state $x$. The functions $\psi_n$ in \eqref{eq:mck_psin}, which constitute $w_0$, represent the contribution of each generation. Thus, in the model where $F(x\leftarrow \xi;a)$ determines both the number and the initial state of offspring, the stable density \eqref{eq:mck_stable_birth} aggregates contributions over all generations.

For the contribution of the initial condition, we compute
\begin{align}
\langle v_0,\hat G(r_0)\rangle_x
&=
\int_{\Omega^d}v_0(y_0,x)\,\hat G(x;r_0)\,dx \nonumber\\
&=
\int_0^\alpha\int_{\Omega^d}v_0(a,y_0,\eta)\,\varphi(a,\eta)\,d\eta\,da,
\label{eq:mck_prefactor_num}
\end{align}
with
\begin{equation}\label{eq:mck_v0_age}
\begin{split}
v_0(a,y_0,x)
:=&
\int_{\Omega^d}v_0(y_0,\xi)
\int_a^\alpha e^{-r_0(\tau-a)}
\int_{\Omega^d}
F(\xi\leftarrow \eta;\tau)\,K(\tau,\eta\leftarrow a,x)\,d\eta\,d\tau\,d\xi.
\end{split}
\end{equation}
Thus $\langle v_0,\hat G(r_0)\rangle_x$ is the direct contribution of the initial age-state distribution $\varphi$ to the descendants of the reference state.

Whenever derivatives of order $k$ are used below, we additionally assume the bound in Proposition~\ref{prop:mck_psi_in_K} with the integrand multiplied by $a^k$; for a full cumulant expansion, we assume the corresponding exponential moment locally in $r$.

In all series manipulations below, we assume absolute rearrangeability with the integrals and dual pairings involved.
This rearrangeability condition is automatic in finite-state models and is satisfied in the usual bounded-kernel settings under norm-dominated convergence.

For the generation-time denominator, we write
\begin{align}
&-
\left\langle
v_0,\frac{d}{dr}\hat\Psi(r)\Big|_{r=r_0}w_0
\right\rangle_x
=
-\int_{\Omega^d}\int_{\Omega^d}
v_0(y_0,x)\,
\left.\frac{d}{dr}\psi(x,\xi;r)\right|_{r=r_0}\,
w_0(\xi,y_0)\,d\xi\,dx \nonumber\\
&=
-v_0(y_0)w_0(y_0)
\sum_{n=2}^{\infty}\sum_{m=1}^{n-1}
\int_{\Omega^d}\int_{\Omega^d}
\psi_{n-m}^{*}(y_0,x;r_0)\,
\left.\frac{d}{dr}\psi(x,\xi;r)\right|_{r=r_0}\,
\psi_{m}(\xi,y_0;r_0)\,d\xi\,dx,
\label{eq:mck_prefactor_den}
\end{align}
where
\begin{equation}\label{eq:mck_dpsi}
-\left.\frac{d}{dr}\psi(x,\xi;r)\right|_{r=r_0}
=
\int_0^\alpha
a\,e^{-r_0a}
\int_{\Omega^d}
F(x\leftarrow \eta;a)\,K(a,\eta\leftarrow 0,\xi)\,d\eta\,da.
\end{equation}

To normalize, set
\begin{equation}\label{eq:mck_En_norm}
\langle v_0,\hat\Psi(r_0)w_0\rangle_x
=
\langle v_0,w_0\rangle_x
=
v_0(y_0)w_0(y_0)\,E_n(y_0),
\end{equation}
where the mean contributing generation number is
\begin{equation}\label{eq:mck_En}
E_n(y_0):=\sum_{n=1}^{\infty}n\,\psi_n(y_0,y_0;r_0).
\end{equation}
Although the individual coefficients may have either sign, \eqref{eq:mck_En_norm}
shows that
\[
E_n(y_0)=\frac{\langle v_0,w_0\rangle_x}{v_0(y_0)w_0(y_0)}>0.
\]
The term ``mean'' is used by analogy with the mean recurrence generation in a
taboo decomposition; it does not assert that the individual coefficients form
a probability mass function.

Dividing \eqref{eq:mck_prefactor_den} by \eqref{eq:mck_En_norm} yields the average generation interval
\[
\bar g_{\mathrm L}
:=
\frac{-1}{E_n(y_0)}
\sum_{n=2}^{\infty}\sum_{m=1}^{n-1}
\int_{\Omega^d}\int_{\Omega^d}
\psi_{n-m}^{*}(y_0,x;r_0)
\left.\frac{d}{dr}\psi(x,\xi;r)\right|_{r=r_0}
\psi_m(\xi,y_0;r_0)\,d\xi\,dx.
\]

Using the interpretations of \eqref{eq:mck_v0_age} and \eqref{eq:mck_adj_birth}, equation \eqref{eq:mck_prefactor_num} yields the direct contribution
\[
\bar R_{\mathrm L}(y_0)
:=
\frac{\langle v_0,\varphi\rangle_{a,x}}{v_0(y_0)},
\qquad
\langle v_0,\varphi\rangle_{a,x}
:=
\int_0^\alpha\int_{\Omega^d}v_0(a,y_0,x)\varphi(a,x)\,dx\,da.
\]

For the cohort quantities below, assume additionally that the hypotheses of Proposition~\ref{prop:mck_psi_in_K} hold at $r=0$ and that the corresponding spectral separation and reference non-resonance assumptions hold for $\hat\Psi(0)$.
We may then define cohort-based quantities in terms of the spectral radius
\[
\lambda_0:=\rho(\hat\Psi(0))
\]
of the next-generation operator $\hat\Psi(0)$.  Let $\bar w_0$ and $\bar v_0$
be positive right and left eigenfunctions of $\hat\Psi(0)$, with eigenvalue
$\lambda_0$, represented by the corresponding normalized $\psi_n$-series, and
write $\bar w(y_0):=\bar w_0(y_0)>0$ and
$\bar v(y_0):=\bar v_0(y_0)>0$.
Thus
\[
E_n^c(y_0)
=\frac{\langle\bar v_0,\bar w_0\rangle_x}
{\bar v(y_0)\bar w(y_0)}>0.
\]
The cohort-based quantities are:
\begin{align}
E_n^c(y_0)
&:=
\sum_{n=1}^{\infty}\frac{n}{\lambda_0^n}\,\psi_n(y_0,y_0;0),\label{eq:mck_Enc}\\
\bar g_0
&:=
\frac{-1}{E_n^c(y_0)}
\sum_{n=2}^{\infty}\frac{1}{\lambda_0^{n+1}}
\sum_{m=1}^{n-1}
\int_{\Omega^d}\int_{\Omega^d}
\psi_{n-m}^{*}(y_0,x;0)
\left.\frac{d}{dr}\psi(x,\xi;r)\right|_{r=0}
\psi_m(\xi,y_0;0)\,d\xi\,dx,\nonumber\\
\bar R_0(y_0)
&:=
\sum_{n=1}^{\infty}\frac{1}{\lambda_0^{n}}
\int_0^\alpha\int_{\Omega^d}\int_{\Omega^d}
\psi_n^*(y_0,\xi;0)
\int_a^\alpha\int_{\Omega^d}
F(\xi\leftarrow \eta;\tau)\,
K(\tau,\eta\leftarrow a,x)\,
\varphi(a,x)\,d\eta\,d\tau\,d\xi\,dx\,da.\nonumber
\end{align}

In this multistate setting, $\lambda_0$ corresponds to the basic (net) reproduction number \cite{inaba2017age}. The average life expectancy adjusted for the population growth rate is
\[
e_0(r_0)
:=
\frac{\int_0^\alpha\int_{\Omega^d}w_0(a,x,y_0)\,dx\,da}
{\int_{\Omega^d}w_0(0,\xi,y_0)\,d\xi}.
\]
As a new demographic indicator, we define the per-generation total contribution of the cohort with initial state $y_0$, denoted by $\Upsilon(y_0,r_0)$, by
\begin{equation}
\Upsilon(y_0,r_0)
:=
\frac{\int_{\Omega^d}w_0(0,\xi,y_0)\,d\xi}{w_0(y_0)E_n(y_0)}.
\end{equation}
At replacement level ($r_0=0$ so $\lambda_0=1$),
\[
E_n(y_0)=E_n^c(y_0),
\qquad
\bar g_{\mathrm L}=\bar g_0,
\qquad
\bar R_{\mathrm L}(y_0)=\bar R_0(y_0).
\]
Thus, at replacement level, the stationary population is characterized by the contribution to descendants with a specific initial state $y_0$, the generation interval, the life expectancy at birth, and the total contribution of the cohort with initial state $y_0$, consistent with the classical McKendrick/Leslie theory. Furthermore, the introduction of the average contributory generation number $E_n^c(y_0)$ together with $\Upsilon(y_0,0)$ provides genealogical resolution beyond earlier models.

\subsection{Other demographic indicators and statistical quantities}

Similarly, the type reproduction number is the direct contribution from an ancestor with the same initial condition $y_0$:
\begin{equation}
T_{y_0}
=
\sum_{n=1}^{\infty}\psi_n(y_0,y_0;0).
\end{equation}
Let
\[
\lambda_0=\rho(\Psi(0)),
\]
where $\rho(\Psi(0))$ denotes the spectral radius of the next-generation operator. 
This quantity can be interpreted as the basic reproduction number, representing the total expected reproductive contribution generated by a single individual.
\begin{corollary}
 By the definition of the basic reproduction number $\lambda_0$,
\begin{equation}
1
=
\sum_{n=1}^{\infty}\frac{1}{\lambda_0^{n}}\psi_n(y_0,y_0;0).
\end{equation}
Hence, if $\lambda_0=1$, then $T_{y_0}=1$ and $r_0=0$. Moreover, by monotonicity of $\rho(\hat\Psi(r))$ in $r$, it follows that $\lambda_0<1$ implies $T_{y_0}<1$ and $r_0<0$.
\end{corollary}

\noindent\textbf{Associated probabilities (definitions).}
We recall the two associated probability measures on $[0,\alpha)$:
\begin{equation}
\begin{aligned}
\mathbb P_{\mathrm L}([0,a);y_0)
&:=
\frac{1}{E_n(y_0)}
\sum_{n=2}^{\infty}\sum_{m=1}^{n-1}
\int_0^a e^{-r_0\tau}
\int_{\Omega^d}\int_{\Omega^d}
\psi_{n-m}^{*}(y_0,x;r_0)\\
&\qquad\times
\int_{\Omega^d}
F(x\leftarrow \eta;\tau)\,
K(\tau,\eta\leftarrow 0,\xi)\,
\psi_m(\xi,y_0;r_0)\,d\eta\,d\xi\,dx\,d\tau,
\end{aligned}
\label{eq:mck_PL}
\end{equation}
\begin{equation}
\begin{aligned}
\mathbb P_{0}([0,a);y_0)
&:=
\frac{1}{E_n^{c}(y_0)}
\sum_{n=2}^{\infty}\frac{1}{\lambda_0^{n+1}}
\sum_{m=1}^{n-1}\int_0^a
\int_{\Omega^d}\int_{\Omega^d}
\psi_{n-m}^{*}(y_0,x;0)\\
&\qquad\times
\int_{\Omega^d}
F(x\leftarrow \eta;\tau)\,
K(\tau,\eta\leftarrow 0,\xi)\,
\psi_m(\xi,y_0;0)\,d\eta\,d\xi\,dx\,d\tau,
\end{aligned}
\label{eq:mck_P0}
\end{equation}
with the normalizing constants
\[
E_n(y_0):=\sum_{n=1}^{\infty}n\,\psi_n(y_0,y_0;r_0),
\qquad
E_n^{c}(y_0):=\sum_{n=1}^{\infty}\frac{n}{\lambda_0^{n}}\psi_n(y_0,y_0;0).
\]
Resummation gives the nonnegative densities
\[
\begin{aligned}
\frac{d\mathbb P_{\mathrm L}}{da}
&=\frac{e^{-r_0a}}{\langle v_0,w_0\rangle_x}
\int v_0(y_0,x)F(x\leftarrow\eta;a)
K(a,\eta\leftarrow0,\xi)w_0(\xi,y_0)\,d\eta\,d\xi\,dx,\\
\frac{d\mathbb P_0}{da}
&=\frac{1}{\lambda_0\langle\bar v_0,\bar w_0\rangle_x}
\int \bar v_0(y_0,x)F(x\leftarrow\eta;a)
K(a,\eta\leftarrow0,\xi)\bar w_0(\xi,y_0)\,d\eta\,d\xi\,dx.
\end{aligned}
\]
Their total masses are one by the corresponding eigen-equations.  Thus these
are positive measures even when their coefficient expansions contain signed
terms.

\medskip
\noindent\textbf{Cumulant expansions under the associated probabilities.}
Write the expectations under \eqref{eq:mck_PL}--\eqref{eq:mck_P0} as
\[
\mathbb E_{y_0}[\cdot]:=\int(\cdot)\,\mathbb P_{\mathrm L}(da;y_0),
\qquad
\mathbb E^c_{y_0}[\cdot]:=\int(\cdot)\,\mathbb P_{0}(da;y_0).
\]
Define the cumulant generating functions
\[
\Theta_{\mathrm L}(t;y_0):=\ln \mathbb E_{y_0}[e^{ta}],
\qquad
\Theta_{0}(t;y_0):=\ln \mathbb E^c_{y_0}[e^{ta}],
\]
and the cumulants $\kappa_k(y_0):=\partial_t^k\Theta_{\mathrm L}(0;y_0)$, $\kappa_k^c(y_0):=\partial_t^k\Theta_0(0;y_0)$ for $k\ge1$.
Then, for $r$ near $r_0$ and $r$ near $0$, respectively,
\begin{equation}\label{eq:mck_cumulant_L}
\begin{aligned}
\ln\!\Biggl(
\frac{\langle v_0,\hat\Psi(r)w_0\rangle_x}{v(y_0)w(y_0)}
\Biggr)
&=
\ln E_n(y_0)+\Theta_{\mathrm L}\bigl(-(r-r_0);y_0\bigr)\\
&=
\ln E_n(y_0)
+
\sum_{k=1}^{\infty}\frac{\kappa_k(y_0)}{k!}(-1)^k(r-r_0)^k,
\end{aligned}
\end{equation}
\begin{equation}\label{eq:mck_cumulant_0}
\begin{aligned}
\ln\!\Biggl(
\frac{\big\langle \bar v_0,\frac{1}{\lambda_0}\hat\Psi(r)\bar w_0\big\rangle_x}
{\bar v(y_0)\bar w(y_0)}
\Biggr)
&=
\ln E_n^c(y_0)+\Theta_0(-r;y_0)\\
&=
\ln E_n^c(y_0)
+
\sum_{k=1}^{\infty}\frac{\kappa_k^c(y_0)}{k!}(-1)^k r^k.
\end{aligned}
\end{equation}

\medskip
\noindent\textbf{Second-order truncation (generation-time statistics).}
Retaining only the first two cumulants in \eqref{eq:mck_cumulant_L}--\eqref{eq:mck_cumulant_0} yields
\begin{equation}
\ln\!\Biggl(
\frac{\langle v_0,\hat\Psi(r)w_0\rangle_x}{v(y_0)w(y_0)}
\Biggr)
=
\ln E_n(y_0)-\bar g_{\mathrm L}(r-r_0)
+\frac{1}{2}\sigma_{\mathrm L}^{2}(r-r_0)^2
+o\!\big((r-r_0)^2\big),
\end{equation}
\begin{equation}
\ln\!\Biggl(
\frac{\big\langle \bar v_0,\frac{1}{\lambda_0}\hat\Psi(r)\bar w_0\big\rangle_x}
{\bar v(y_0)\bar w(y_0)}
\Biggr)
=
\ln E_n^c(y_0)-\bar g_{0}\,r
+\frac{1}{2}\sigma_{0}^{2}r^2
+o(r^2),
\end{equation}
where the generation-time mean and variance under the associated probabilities are
\begin{align*}
\bar g_{\mathrm L}&:=\kappa_1(y_0)=\mathbb E[a],\qquad
\sigma_{\mathrm L}^{2}:=\kappa_2(y_0)=\mathbb V_{y_0}(a),\\
\bar g_{0}&:=\kappa_1^c(y_0)=\mathbb E^c[a],\qquad
\sigma_{0}^{2}:=\kappa_2^c(y_0)=\mathbb V^c_{y_0}(a).
\end{align*}

\begin{remark}
\rm
Accordingly, any representative demographic indicator derived from the multi-state McKendrick equation inherently reflects the entire sequence of intergenerational transitions and cannot be characterized solely by cohort-based quantities, as in the classical McKendrick or Leslie models.
\end{remark}

\subsection{Consistency of the reference-point normalization with the Euler--Lotka equation}
\label{app:EL-ref0}

In the main text we determine the eigenvalue by normalizing the eigenfunctions so that their values at the reference point equal $1$. For the classical one-state McKendrick--von Foerster model, taking the reference point at age $0$ shows that this normalization reproduces the Euler--Lotka equation.

Let $\ell(a)$ be the survival function and $\beta(a)$ the fertility rate, and write $\lambda=e^{r}>0$. In the classical theory, the stable age density has the form
\[
w(a)=c_w e^{-ra}\ell(a),
\qquad a\ge0,
\]
so that its value at age $0$ is $w(0)=c_w$. Hence the newborn production is
\begin{equation}\label{eq:mck_w0_EL}
w(0)=\int_{0}^{\infty}\beta(a)w(a)\,da
=
c_w\int_0^\infty e^{-ra}\beta(a)\ell(a)\,da.
\end{equation}
Therefore, imposing the reference-point normalization $w(0)=1$ is equivalent to
\[
\int_0^\infty e^{-ra}\beta(a)\ell(a)\,da=1,
\]
which is exactly the Euler--Lotka equation.

Likewise, under the same reference-point viewpoint, the reproductive value at age $0$ is a constant multiple of the Euler--Lotka functional:
\[
v(0)=c_v\int_0^\infty e^{-ra}\beta(a)\ell(a)\,da,
\]
for a constant $c_v>0$ depending only on the chosen normalization. Recalling the scalar identity \eqref{eq:key_relation_cD},
\[
\frac{1}{c_w}c(\lambda)D(\lambda)=1-\frac{w(0)}{c_w},
\]
we see that, in the classical McKendrick theory, choosing age $0$ as the reference point amounts precisely to the reference-point eigenstructure: the eigenvalue is recovered by fixing the eigenfunction value at the reference point.
\section{Discussion}
\label{sec:discussion}

This paper develops a determinant-free framework for describing the dominant
eigenstructure of positive Fredholm operators through a reference-point
construction.  The main result is a kernel-level recursion whose normalized
series represents the leading right and left eigenfunctions under two explicit
conditions: dominant spectral separation and scalar non-resonance for the
chosen reference pair.

The reference-point operator transfers the renewal intuition behind discrete
taboo decompositions to continuous kernels.  In discrete Markov chains, taboo
probabilities decompose paths according to visits to a distinguished state. In
the present setting the analogous construction is made on kernels.  The
resulting coefficient growth is controlled by the reference-point generating
identity and the \(\kappa\)-dependent non-resonance condition.

The kernel class \(\mathbb K\) makes the reference-point
recursion meaningful, but it does not by itself imply convergence at the
Perron root.  The result is therefore a conditional representation theorem for
operators whose dominant spectral decomposition and reference denominator can
be checked.

For discrete-time IPMs, the expansion expresses the stable distribution and
the reproductive value through iterated kernels, which collect multistep
transition and reproductive contributions.  For multi-state McKendrick
equations, the age-zero renewal kernel reduces the dominant-root problem to an
Euler--Lotka-type equation at the kernel level.  These applications are meant
to illustrate the operator-theoretic construction; they do not require a
separate biological modelling assumption beyond the stated kernel and
spectral hypotheses.
In biological terms, the genealogical expansion decomposes the stable
population and reproductive value into contributions indexed by the number of
successive reproductive or transition generations represented by the iterated
kernels.

Several questions remain.  It would be useful to replace point references by
reference sets of positive measure, which may lead to more robust numerical
schemes and a more direct applied interpretation.  Another direction is to
identify model-specific conditions under which the dominant spectral
separation and reference non-resonance assumptions can be verified directly
from the ingredients of empirical integral projection models.

Overall, the reference-point construction provides a new representation of
dominant eigenstructures that applies to both IPMs and age-structured renewal
models without determinant expansions or discretization, while retaining a
direct decomposition of the contributions that build the leading eigensystem.
\section*{Acknowledgements}
The authors are grateful to Kumiko Oizumi, Shin Oizumi, Ko Oizumi, and Hiroko Oizumi for their support and encouragement. The authors also thank Hisashi Inaba for valuable advice and Youichi Enatsu for helpful discussions. The authors pay tribute to the late Professor Nobuhiko Fuji, who passed away in December 2025, and gratefully acknowledge the many fruitful discussions and the mathematical knowledge he shared.

\section*{Statements and Declarations}

\paragraph{Funding.}
This work was supported by Health, Labour and Welfare Sciences Research Grants from the Ministry of Health, Labour and Welfare of Japan (Grant Number JPMH26AA2009; project title: ``Population and household projections and social restructuring focusing on diversifying household structures''). Yuki Chino acknowledges support from the NSTC grant 111-2115-M-A49-009-MY3.

\paragraph{Competing interests.}
The authors declare that they have no competing interests relevant to the content of this article.

\end{document}